\documentclass[10pt,pra,letterpaper,typeset,superscriptaddress,twocolumn,showpacs, floatfix]{revtex4-2}

\usepackage[ruled,vlined]{algorithm2e}

\usepackage{amsfonts}
\usepackage{amsmath}
\usepackage{amssymb}
\usepackage{graphicx}
\usepackage{xcolor}

\usepackage{physics}
\usepackage{comment}
\usepackage[justification=justified,singlelinecheck=false]{caption} 
 
\newcommand{\jcz}[1]{{\color{blue}(JCz: \textit{#1})}}

\newcommand{\vdga}[1]{{\color{purple}(VGA: \textit{#1})}}

\usepackage{ragged2e} % para justificar texto

\usepackage{subcaption}
\usepackage{amsthm}
\theoremstyle{plain}
\usepackage{rotating} 
\DeclareMathOperator{\CNOT}{CNOT}
\usepackage{hyperref}
\usepackage{thmtools} 
\usepackage{thm-restate}

\newtheorem{defn}{Definition}

\newtheorem{lem}{Lemma}

\newtheorem{obs}{Observation}

\begin{document}

%\preprint{APS}

\title{Efficient state estimation on quantum processors}

%More alternatives
%\title{Estimating arbitrary multi-qubit pure states without entangled measurements}
%\title{Local tomography for scalable quantum computers}
%\title{Minimal circuits for multi-qubit state estimation}

\author{V. Gonz{\'a}lez Avella}
%\email{victor.gonzalez.avella@ua.cl}
\affiliation{Departamento de F{\'i}sica, Facultad de Ciencias B{\'a}sicas,
Universidad de Antofagasta, Casilla 170, Antofagasta, Chile}
%\orcid{0000-0003-2633-6146}

\author{A. Vega Vargas}

\affiliation{Departamento de Matemática y Ciencia de la Computación, Universidad de Santiago de Chile, Santiago, Chile}

\author{T. Merlo Vergara}

\affiliation{Instituto de F\'isica, Pontificia Universidad Cat\'olica de Chile, Casilla 306, Santiago, Chile}

\author{K. de la Ossa Doria}
\affiliation{Departamento de F{\'i}sica, Facultad de Ciencias B{\'a}sicas,
Universidad de Antofagasta, Casilla 170, Antofagasta, Chile}

\author{J. Czartowski}

\affiliation{Faculty of Physics, Astronomy and Applied Computer Science, Jagiellonian University, ul. Łojasiewicza 11, 30-348 Kraków, Poland}

\affiliation{School of Physical and Mathematical Sciences, Nanyang Technological University,
21 Nanyang Link, 637371 Singapore, Republic of Singapore}

\author{D. Main}

\affiliation{Department of Physics, University of Oxford, Clarendon Laboratory,
Parks Road, Oxford OX1 3PU, United Kingdom}

\author{G. Araneda}

\affiliation{Department of Physics, University of Oxford, Clarendon Laboratory,
Parks Road, Oxford OX1 3PU, United Kingdom}

\author{A. Delgado}

\affiliation{Departamento de F\'isica e Instituto Milenio de Investigaci\'on en \'Optica, Universidad de Concepci\'on, Casilla 160-C, Concepci\'on, Chile}

\author{D. Goyeneche}

\affiliation{Instituto de F\'isica, Pontificia Universidad Cat\'olica de Chile, Casilla 306, Santiago, Chile}

\date{\today}

\begin{abstract}
%IA improved version (just grammar and style)

%We present two scalable and entanglement-free methods for estimating the collective state of a $n$-qubit quantum computer. The first method consists of a fixed set of five quantum circuits -- regardless of the number of qubits -- that avoid the use of entanglement as a measurement resource, relying instead on classical communication between selected pairs of qubits. The second method requires only $2n + 1$ circuits, each of which applies a single local gate to only one of the $n$ qubits during the measurement stage. Unlike traditional estimation methods, our approaches do not require any costly post-processing procedure to estimate a quantum state, enabling scalability to relatively large system sizes. We experimentally validated both methods on freely available IBM quantum processors, achieving fidelity above 0.93 for systems of 12 qubits on real hardware. We validated our results by estimating the 4-qubit entangled state of two remote ion-trap quantum processors, demonstrating that the optimized \(2n+1\) tomographic scheme achieves estimates consistent with standard  methods while using exponentially fewer measurements.

We present two scalable, entanglement-free methods for estimating the collective state of an $n$-qubit quantum computer. The first method consists of a fixed set of five quantum circuits—independent of the number of qubits—that avoid the use of entanglement as a measurement resource, relying instead on classical communication between selected pairs of qubits. The second method requires $2n + 1$ circuits, each applying a single local gate to one of the $n$ qubits during the measurement stage. Unlike traditional estimation methods, our approaches yield an analytical reconstruction formula that eliminates the need for costly post-processing to estimate the quantum state, thereby enabling scalability to relatively large system sizes. We experimentally compare both methods on freely available IBM quantum processors and observe how state estimation varies with increasing numbers of qubits and shots. We further validate our results by estimating a 4-qubit entangled state distributed across two remote ion-trap quantum processors, demonstrating that the optimized $2n + 1$ tomographic scheme achieves estimates consistent with standard methods while requiring exponentially fewer measurements.

%We also validated our results by estimating the 4-qubit entangled state of two remote ion-trap small quantum processors, outperforming a widely used tomographic method. {\bf(Esto ultimo no es claro, no aparece en la seccion correspondiente. En dicha seccion tampoco hay una discusion de los resultados obtenidos)}

%Human version
%We introduce two minimalist and scalable quantum state estimation protocols for reconstructing the global state of $n$-qubit quantum computers. The first protocol requires $2n+1$ circuits, each of them composed of a single local gate applied to a single qubit in the measurement stage. The second one is a fixed set of 5 quantum circuits, for any number of qubits, that does not require entanglement as a resource in the measurement stage but considers classical communication between some pairs of qubits. As an advantage with respect to traditional tomographic methods, our approach does not require postprocessing to estimate a quantum state, allowing us to scale up to a relatively large number of qubits. Both methods are experimentally tested in the lab for 4 qubits and on IBM quantum computers, achieving fidelities above 0.9 for 12 qubits on a real device.
\end{abstract}


\maketitle

\section{Introduction}
    Quantum hardware, such as quantum computers and simulators, has grown steadily in size and complexity over the past few decades. Currently, superconducting qubits~\cite{Wendin_2017,Kjaergaard_2020}, trapped ions~\cite{Bruzewicz_2019}, cold optical lattices~\cite{Cornish_2024}, and Rydberg atom arrays~\cite{Cong_2022} have been used to demonstrate coherent control from tens to hundreds of qubits and gate layers, with thousands of qubits expected to be achieved in the near future, a regime where quantum hardware simulation on classical hardware becomes infeasible.
    
    In this scenario, techniques for efficient and accurate characterization and benchmarking of multi-qubit quantum hardware become essential. Estimation of quantum states and processes is achieved by post-processing through statistical inference of information acquired through measurements on a finite ensemble~\cite{Czerwinski_2022,Bisio_2009}. Universal estimation techniques, which reconstruct all pure or mixed states, suffer from algorithmic complexity that scales exponentially with the number of qubits~\cite{Anshu_2024}. It is possible to reduce complexity by resorting to a priori information. For example, five bases~\cite{goyeneche2015five} have been shown to allow the estimation of pure $n$-qubit states without the use of post-processing based on statistical inference, which involves a total number of measurement outcomes that scales linearly with the dimension.
    
    Measurements are realized in quantum hardware through quantum circuits, which can typically involve the application of many entangling gates, such as the controlled NOT gate. These gates have a high error rate that affects the accuracy of the estimation, and therefore it is desirable that during the current generation of noisy intermediate-scale quantum (NISQ) hardware~\cite{Preskill_2018,Chen_2023}, attempts are made to estimate states and processes solely by local means. In this direction, it has recently been shown that it is possible to estimate pure $n$-qubit states through strictly local bases~\cite{Pereira_2022}, although they scale in quantity linearly with the number of qubits.
    
    In recent decades, several quantum state estimation methods have been proposed, each of them solving one or more fundamental aspects related to the recovery of complete information of a quantum system. Standard quantum tomography for an ensemble of $n$ qubits~\cite{James_2001} is based on measuring all %$4^n$ 
    chains of local Pauli operators, %and the identity operator \dg{why the identity? this is just normalization} f
    followed by applying maximum likelihood estimation. %to the statistical frequencies of all 
    % $6^n$ 
    %\bl{$4^n + 2^n$. (?)}\jcz{Discussion hidden} 
    % \vdga{is this $6^n$ or $8^n$} \dg{for each string of Pauli operator there are $2^n$ rank-one projectors, so the total number of outcomes could be, innocently, estimated as $2^n(4^n-1)=2^{3n}-2^n=8^n-2^n$. However, many of those strings mutually commute, so the number of required strings of rank one projectors is reduced. From mutually unbiased bases theory, we know that the $4^n-1$ Pauli strings can be grouped into $2^n+1$ sets of $2^n-1$ commuting operators each. So, there are $2^n(2^n+1)=4^n+2^n$ essential measurement outcomes associated to Pauli strings, right?} \jcz{Sounds about right to me!}\vdga{Is this Number is equal to MUB?, look the next sentence} \dg{Yes. Indeed, $d+1$ maximal set of $d-1$ Pauli strings produce maximal sets of MUB in dimension $d$. Each basis is given by the unique eigenvectors basis that is common to each set of $d-1$ commuting operators.}. 
    The number of required measurement outcomes can be reduced by considering a symmetric, informationally complete, positive-operator-valued measure~\cite{Renes_2004} %on each qubit~\cite{Stricker_2022} or to $2^n(2^n+1)$ by using a set of $2^n+1$ 
    or mutually unbiased bases~\cite{Wootters_1989}, at the cost of requiring entanglement as a resource at the measurement stage~\cite{wiesniak2011entanglement,czartowski2018entanglement}. Furthermore, quantum compressed sensing~\cite{gross2010quantum} can estimate mixed states from a selected subset of Pauli chains, achieving an exponential reduction in the number of measurements as a function of the number of qubits $n$. However, all the methods mentioned above require resources that scale exponentially with $n$, making their practical implementation unfeasible. Having quantum states with a matrix product state structure~\cite{Cramer_2010} or permutation invariance~\cite{Toth_2010} also leads to a significant reduction in the total number of measurement results when using local measurements. Although the last two methods do not provide an explicit reconstruction procedure, they lead to an advantageous reduction of the classical computational cost of the post-processing stage. With the aim of reconstructing pure quantum states, for instance, the collective state of a quantum computer, it has been shown that pure $n$-qubit states can be accurately estimated using local $mn+1$ measurement bases (with $m \geq 2$) and without the need for post-processing based on statistical inference methods~\cite{Pereira_2022}. However, this approach requires solving an exponentially growing number of times a linear system of equations, scaling with the Hilbert space dimension $2^n$. Consequently, this method does not exhibit favorable scalability for a relatively large number of qubits.

% \begin{figure}[htbp]
\begin{figure}[t]
    \centering
    \includegraphics[width=0.48\textwidth]{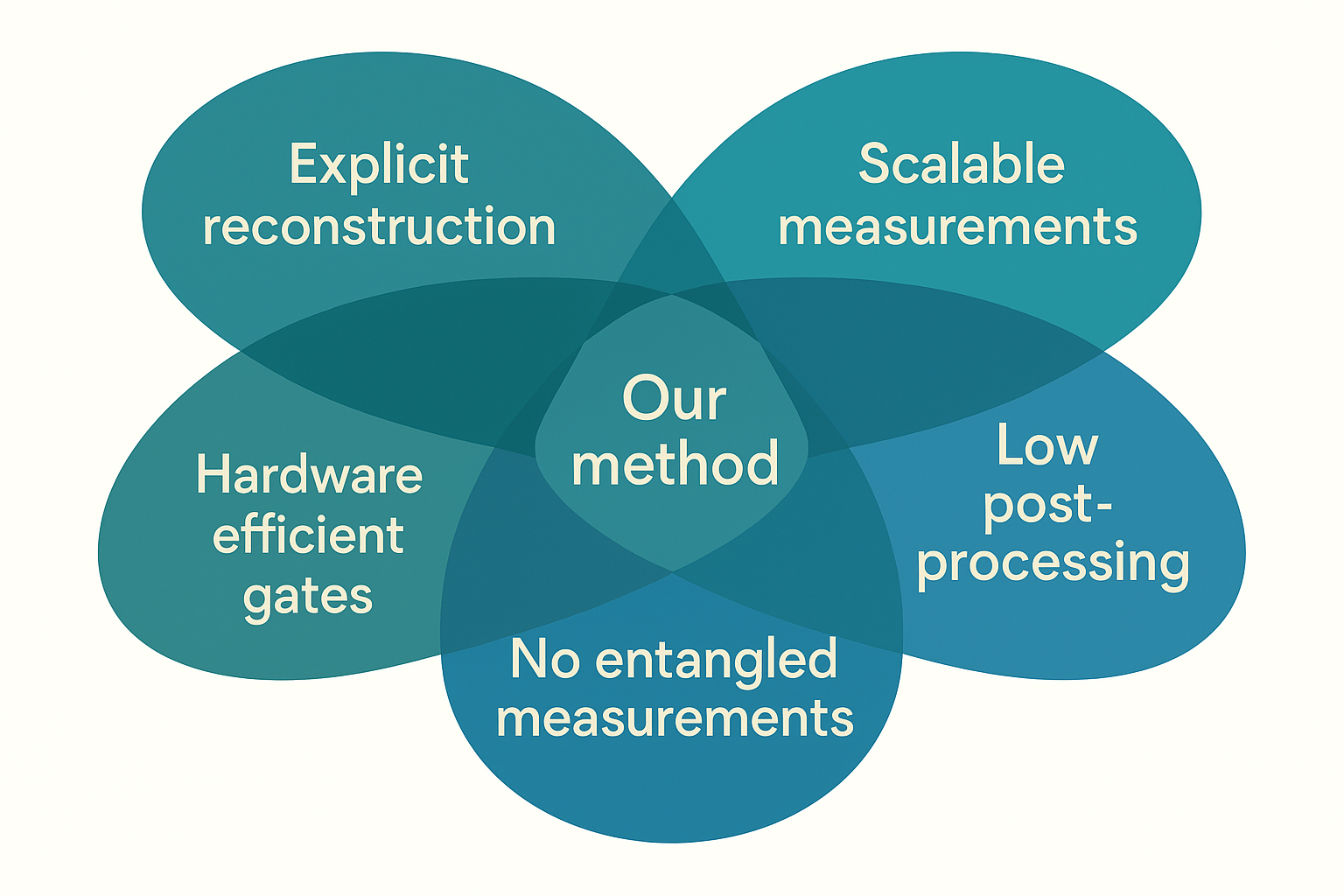}
    \caption{\justifying Schematic representation of the key ingredients addressed by our proposal. Some highly desirable properties of a tomographic method include an explicit reconstruction formula~\cite{scott2006tight,goyeneche2015five,vargas2024near,james2001measurement}, hardware-efficient gates~\cite{gross2010quantum,Pereira_2022,guo2024quantum}, a scalable number of measurements~\cite{goyeneche2015five,vargas2024near,Pereira_2022}, low post-processing cost to reconstruct a positive semidefinite operator~\cite{kaznady2008quantum,goyeneche2015five}, and the absence of entanglement in the measurement stage~\cite{james2001measurement,cramer2010efficient,Pereira_2022}.
While existing methods typically exhibit only a subset of these desirable features, our approach stands out as a powerful and comprehensive method combining analytical simplicity with experimental scalability. This comprehensive optimization represents a significant advance in multi-qubit quantum state tomography, as demonstrated by the experimental results presented in Sections~\ref{sec:IBM} and~\ref{sec:Oxford}.}
    \label{fig:venn_tomography}
\end{figure}

This work resolves a long-standing challenge in quantum pure state estimation by providing scalable, entanglement-free, and computationally efficient methods that remains accurate for large systems. Our estimation methods use separable measurement bases, one of them considering a fixed number of five bases and the other one scaling linearly with the number of qubits. Also, they are implementable by using hardware-efficient quantum gates and provide an explicit reconstruction formula that does not require computationally expensive post-processing to estimate an unknown quantum pure state (see Fig. \ref{fig:venn_tomography}). In particular, our scalable and robust methods do not require solving a system of equations at any stage—neither during the reconstruction process nor in the post-processing phase. As a result, we achieved unprecedented estimation fidelity in the experimental characterization of genuinely entangled quantum states of more than 10 qubits, using freely available IBM quantum computers. We also validate our results by estimating the 4-qubit entangled state of two remote two-qubit ion-trap quantum processors. 
    
In addition, we present a method that improves upon the five-basis tomographic scheme~\cite{goyeneche2015five} used to estimate arbitrary $d$-dimensional quantum states. In the case of an $n$-qubit system, two of these five bases are implemented through quantum circuits that require a large number of two-qubit entangling gates, which limits the applicability of the method to small system sizes and reduces the estimation fidelity. We introduce a new set of five bases whose implementation relies solely on local operations assisted by classical communication (LOCC), thereby extending the applicability of the method to larger numbers of qubits and improving the estimation fidelity.

    %a previously published tomographic method, which requires five measurement bases to reconstruct arbitrary 
    %$d$ -dimensional quantum states., see  Goyeneche et al.~\cite{goyeneche2015five}. %The improvement is appreciated when reconstructing $n$ qubit systems, where the original five bases require entanglement in the measurement stage, whereas the new set of five bases considers local operations assisted by classical communication (LOCC). %Then, the total number of measurement outputs scales as $(2n+1)\cdot2^n$. The second method focuses on reducing the total number of measurement outcomes by considering only local measurements by developing an estimation protocol for pure states based on $4$ measurement bases, along with the computational basis, where two of them are strictly local and the remaining two require local measurements and classical communications. Thus, the total number of measurement results scales favorably as $5\cdot2^n$. 
    
The measurement bases for both methods are generated through simple constructive procedures, and the state estimate is obtained via direct calculations, thereby avoiding post-processing based on statistical inference methods whose computational cost becomes prohibitive as the number of qubits increases. Moreover, avoiding controlled gates in the generation of the measurement bases leads to higher estimation fidelity.
    
    %Existing literature:
    %Along the last decades, several tomographic methods were introduced, each of them resolving one or more fundamental aspects related to retrieving the full information of a quantum system 
    
    %Despite the great value of the above mentioned contributions, multi-qubit quantum state tomography is so challenging that requires the consideration of a method that solves all the above posed characteristics. However, until now such a method is not know in the literature.
    
    %In this work, we show two methods that successfully overcome all the above mentioned aspects, achieving unprecedented experimental fidelity reconstruction of quantum states more than 10-qubits of genuinely entangled quantum states.

    %In this paper, we show that it is possible to estimate pure $n$-qubit states with five measurement bases that are generated without entangled gates. Three bases are local, while two are created using single-qubit gates and classical communication.

% \section*{Plan for theory}

% \jcz{

% \begin{enumerate}
%     \item Basic notions
%     \begin{itemize}
%         \item System under consideration -- $n$ qubits, definition of relative phases between components.
%         \item Polarization identity -- quick demonstration how it can be applied to a single qubit
%         \item Hypercubic graph, connection between components and relative phases of a state, and polarization identity.
%     \end{itemize}
% \end{enumerate}}

\section{Theoretical background}

    % \jcz{Loose thoughts, may not fit the final version.}
    % \vdga{I've introduced some modifications, I'll try to describe both methods as follows}

    % \rd{State definition}
     
    We consider the problem of estimating the pure states of an $n$-qubit system. The unknown quantum state $\ket{\psi}$ belongs to the composite Hilbert space $\mathcal{H}_2^{\otimes n}$ and can be represented in the computational basis as
    %\begin{align}\label{eq:state}
    %    \ket{\psi} & = 
    %    \sum_{\vb{j}=0}^{2^n-1} a_{\vb{j}}\ket{\vb{j}} = 
    %    \sum_{\vb{j}=0}^{2^n-1}\abs{a_{\vb{j}}}e^{i\phi_{\vb{j}}}\ket{\vb{j}}.
    %\end{align}
    \begin{align}\label{eq:state}
        \ket{\psi} & = 
        \sum_{\vb{j}=0}^{2^n-1} a_{\vb{j}}\ket{\vb{j}},
    \end{align}
    where $a_{\vb{j}}=\abs{a_{\vb{j}}}e^{i\phi_{\vb{j}}}$ are complex probability amplitudes, $\sum_{j}|a_j|^2=1$, and the state $\ket{\vb{j}}$ denotes a member of the computational basis. The integer label $\vb{j}=0,\dots 2^n-1$ is connected to the labeling of the individual qubit states $\ket{j_i}$ with $j_i=\{0,1\}$ for $i=0,\dots,n$ through the binary representation $\vb{j}=\sum_{i=0}^{n-1} j_i 2^i$, that is, $\ket{\mathbf{j}}=\ket{j_{0}\dots j_{n-1}}$, where the tensor products have been omitted for the sake of clarity.

    A key role in our method is played by the \emph{polarization identity}~\cite{schechter1996handbook}. According to this, for any two complex numbers $a_0,a_1\in\mathbb{C}$, it holds that    \begin{equation}\label{eq:polar_iden}
        4 a_0 a_1^* = \sum_{k=0}^3i^k\abs{a_0+i^ka_1}^2.
    \end{equation}
    Thus, the phase difference $\phi_0-\phi_1$ between $a_0$ and $a_1$ can be evaluated from a combination of squares of absolute values. The five-basis tomographic method is based on this mathematical identity~\cite{goyeneche2015five}, which establishes a direct connection between physically accessible information and entries of the target quantum state. To illustrate this, let us consider the simplest case of a single qubit prepared in a pure state $\ket{\psi}=a_0|0\rangle+a_1|1\rangle$ with $a_0\geq0$, $a_1=|a_1|e^{i\phi_1}$, and $|a_0|^2+|a_1|^2=1$. The first measurement basis is the computational basis that allows us to estimate the absolute values $\abs{a_0}$ and $\abs{a_1}$. Second, the complex phase $e^{i\phi_1}$ can be estimated through the polarization identity \eqref{eq:polar_iden} as follows:
    \begin{eqnarray}\label{pol}
                4a_0a_1^* \!&=&\!4 \abs{a_0}\abs{a_1}e^{-i\phi_1}\nonumber\\ \!&=&\! \abs{a_0\! +\! a_1}^2\!+\! \abs{a_0\! -\! a_1}^2 + i\abs{a_0\! +\! i a_1}^2 - i\abs{a_0\! - i a_1}^2\nonumber \\
                \!& =&\! 2\qty(
                \abs{\ip{+}{\psi}}^2
                -\abs{\ip{-}{\psi}}^2
                +i\abs{\ip{+i}{\psi}}^2
                -i\abs{\ip{-i}{\psi}}^2),\nonumber\\
            \end{eqnarray}
    where $\ket{\pm}=(\ket{0}\pm\ket{1})/\sqrt{2}$ and $\ket{\pm i}=(\ket{0}\pm i\ket{1})/\sqrt{2}$ are the eigenvector bases of Pauli matrices $X=|0\rangle\langle1|+|1\rangle\langle0|$ and $Y=-i|0\rangle\langle1|+i|1\rangle\langle0|$, respectively. Note from (\ref{pol}) that the statistics produced by the above three measurement bases, that is, $\{\ket{0},\ket{1}\}$, $\{\ket{+},\ket{-}\}$ and $\{\ket{+i},\ket{-i}\}$, allow us to univocally reconstruct the single-qubit state $|\psi\rangle$. Moreover, these three measurement bases determine a maximal set of \emph{mutually unbiased bases}, which are informationally complete for any single-qubit density matrix~\cite{ivonovic1981geometrical}. 

    The above approach based on the polarization identity can be extended to the case of a $n$-qubit system prepared in an unknown quantum pure state $\ket{\psi}$ Eq.~(\ref{eq:state}). The amplitudes $|a_{\vb{j}}|$ are determined through measurements carried out in the computational basis, whereas the complex phases $e^{i\phi_{\vb{j}}}$ are estimated by solving a set of polarization identities as     \begin{equation}\label{pol_identity}
                4a_{\vb{j}}a_{\vb{j}'}^* = \sum_{\ell=0}^3i^{\ell}\abs{a_{\vb{j}} +i^{\ell} a_{\vb{j}'}}^2
                = \sum_{\ell=0}^3i^{\ell}\abs{(\langle \vb{j}| +i^{\ell} \langle \vb{j}'|)|\psi\rangle}^2.
    \end{equation}
    This system of equations is overdetermined and thus there are several ways to find the complex phases $e^{i\phi_{\vb{j}}}$. 
For example, choosing $\vb{j}' = \vb{j} + 1 \operatorname{mod}{2^n}$ leads to the following set of $4\times2^n$ states  
\begin{equation}
\label{5bases2015}
% \{|\vb{j}\rangle+i^{\ell}|\text{succ}(\vb{j})\rangle\}
\qty{\ket{\vb{j}} + i^\ell \ket{\vb{j}'}}
_{\vb{j}=\qty{0,1}^n}.
\end{equation}

% \dg{I think that here we don't need to mention $\operatorname{mod}{2^n}$, in the ket} 
These states can be grouped into four orthonormal bases $\mathfrak{B}_k$ ($k=1,\dots,4$), which, together with the computational basis $\mathfrak{B}0$, constitute an informationally complete set for pure quantum states. This is because the coefficients of the pure state $|\psi\rangle$ in the computational basis, namely the set ${a{\vb{j}}}$, can be estimated by recursively solving the system of equations~(\ref{pol_identity}) for $\vb{j}'=\vb{j}+1$. These five bases form the core of the five-basis tomographic method~\cite{goyeneche2015five}. Despite the success of this method for reconstructing $d$-dimensional quantum states, these five measurement bases cannot be efficiently implemented in a quantum computer. In fact, the three measurement bases $\mathfrak{B}_0$, $\mathfrak{B}_1$, and $\mathfrak{B}_2$ are fully separable for any number of qubits, but bases $\mathfrak{B}_3$ and $\mathfrak{B}_4$ require the implementation of a sequence of CNOT gates that rapidly increases with $n$. This naturally leads to the question of whether it is possible to estimate pure quantum states via the polarization identity without entangled measurement bases.

%Thus, a key question arises: \medskip

%\textit{Is it possible to reconstruct pure quantum states via the polarization identities (\ref{pol_identity}) %without considering entanglement resources in the measurement bases?}
%\medskip

\begin{figure}[t!]
        \centering
        \begin{subfigure}{0.235\textwidth}
                    \includegraphics[width=\textwidth]{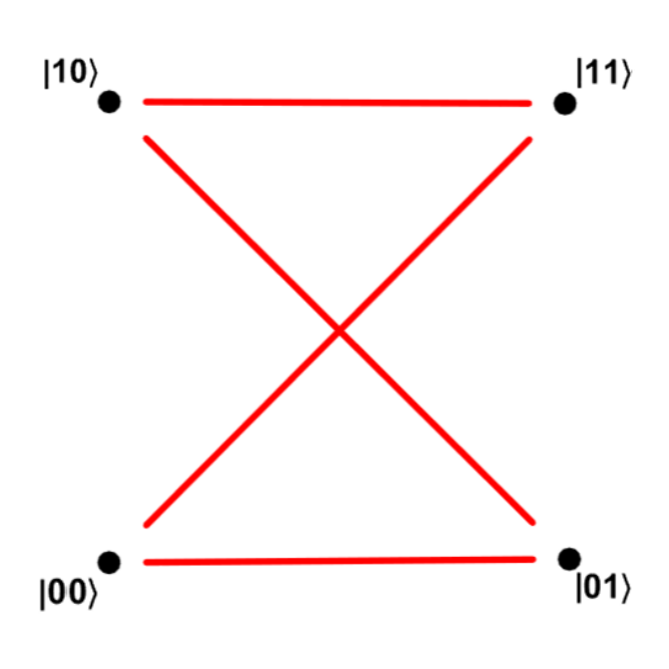}
                    \caption{$\mathcal{G}^{ent}_{2}$}
                    \label{fig:graph2q_a}
        \end{subfigure}
         \begin{subfigure}{0.235\textwidth}
                    \includegraphics[width=\textwidth]{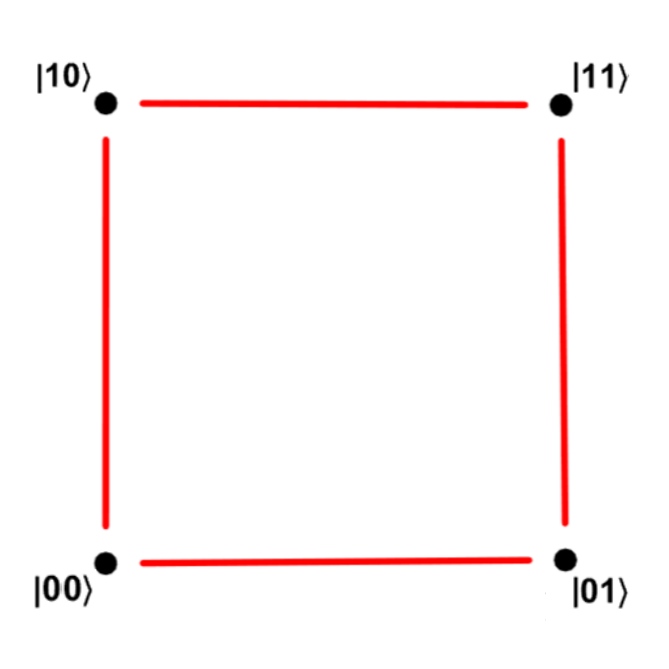}
                    \caption{$\mathcal{G}^{sep}_{2}$}
                    \label{fig:graph2q_b}
        \end{subfigure}
        \caption{\justifying Graphs $\mathcal{G}^{ent}_{2}$ and $\mathcal{G}^{ent}_{2}$ for a $2$-qubit system associated to  (a) entangled and (b) separable measurements, respectively. Note that crossing edges ($\slash$ and \textbackslash) represent entangled measurement states, whereas parallel edges~(\text{---}~and~$|$) represent separable measurement states. In general, edges in the graph $\mathcal{G}$ correspond to separable measurement states if and only if they are also edges of the hypercubic graph.}
        \label{fig:2-qubitexample}
    \end{figure}

\begin{figure}[t!]
        \centering
        \includegraphics[width=0.8\linewidth]{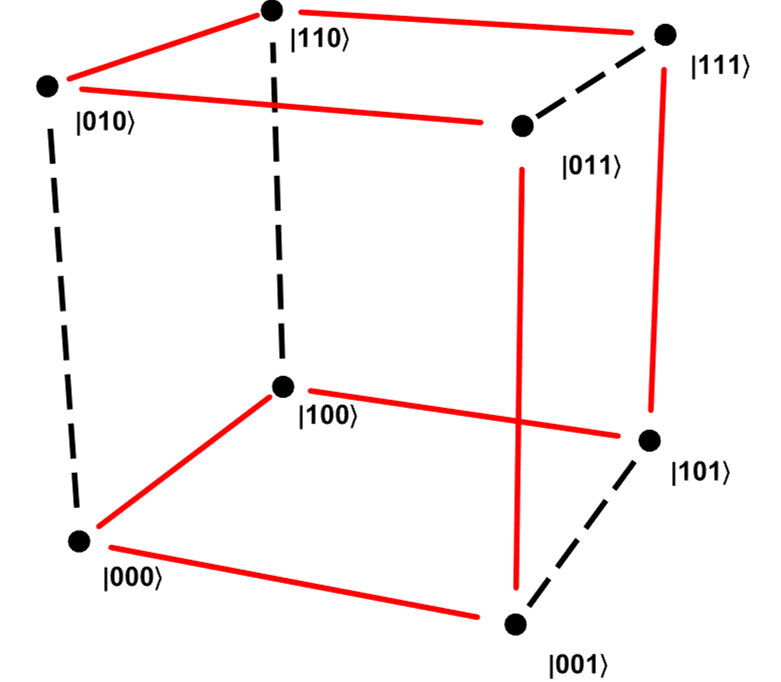}
        \caption{\justifying (Color online) Degree $3$ hypercubic graph with a hamiltonian cycle highlighted
        % \jcz{For BnW printing we should make the remaining black edges dashed... agree?}
        }
        \label{fig:hypercube_3q}  %%%%
    \end{figure}

In this work, we provide a positive answer to this question in two different ways. First, we realize that a rigid rotation of the five bases (\ref{5bases2015}) generates a set of measurement bases that does not contain entangled vectors, for any number of qubits; see Section \ref{sec:5bases}. However, two of these five bases cannot be written as a tensor product of unitary matrices, thus requiring classical communication to be implemented in practice. For this reason, we introduce a second set of measurements consisting of $2n+1$ orthonormal bases that are informationally complete for $n$-qubit pure states, each constructed as a tensor product and requiring no classical communication (see Section~\ref{sec:2n+1_basis}). Furthermore, we demonstrate that our $2n+1$ bases represent the maximal possible number of tensor product bases that can be defined through the polarization identity (\ref{pol_identity}). 

% \dg{Victor, how was this proof?}\vdga{No idea how to do it.}\jcz{Consider the following -- from the graph standpoint, any local measurement induces at most 2-point connections, and additionally, a single measurement is insufficient to recover both real and imaginary part of the state; thus, one needs at least two local measurements per each reduction of the number of graphs connected components in half. As such, one needs at least $2n$ product measurements to create a connected graph. $+1$ is necessary to obtain amplitudes... but I am not so sure if it is absolutely necessary either for generic states. However, using it allows to identify the nodes which "are not there", ie. have zero amplitudes, and thus require a different treatment} \dg{Old comment, sorry. This was already implemented in Section IV}.

Before introducing the two estimation methods, note that there are $2^n(2^n-1)$ possible choices for a polarization identity (\ref{pol_identity}). A suitable subselection of them is required for an estimation method. To provide a simple visualization of this choice, we define an undirected graph $\mathcal{G}(V,E)$, where each vertex $v_{\vb{j}}\in V$ is associated to a coefficient $a_{\vb{j}}$ from the decomposition (\ref{eq:state}) and an edge $e_{\vb{j},\vb{j}'}\in E$ implies that a polarization identity (\ref{pol_identity}) has been considered for the pair of indices $\{\vb{j},\vb{j}'\}$. Additionally, we associate with each vertex $v_{\vb{j}}$ the weight $|a_{\vb{j}}|^2$. For example, the set of polarization identities considered in (\ref{5bases2015}) produces the $2^n$-cycle graph, for any number of qubits $n$. A cycle graph consists of a single closed chain of vertices, where every vertex is connected to 2 edges. In addition, it is simple to note that any subset of equations (\ref{eq:polar_iden}) associated with a connected graph has a unique solution. Indeed, the system of equations can be solved recursively. 

\section{Five disentangled bases}\label{sec:5bases}

Let us start our construction by recalling the five entangled bases $\mathfrak{B}_k$ from \cite{goyeneche2015five} for a 2-qubit system, given by
$$\mathfrak{B}_0=\{|00\rangle,|01\rangle,|10\rangle,|11\rangle\},\mathfrak{B}_1=\{|0,\pm\rangle,|1,\pm\rangle\},$$
$$\mathfrak{B}_2=\{|0,\pm i\rangle,|1,\pm i\rangle\},\mathfrak{B}_3=\{|00\rangle\pm|11\rangle,|01\rangle\pm|10\rangle\},$$
$$\mathfrak{B}_4=\{|00\rangle\pm i|11\rangle,|01\rangle\pm i|10\rangle\},$$
where normalization factors are omitted for simplicity. Note that bases $\mathfrak{B}_3$ and $\mathfrak{B}_4$ are formed by maximally entangled states. In particular, $\mathfrak{B}_3$ is the Bell basis. 

The graph $\mathcal{G}^{ent}_{2}$ associated with the bases above is shown in Fig.~\ref{fig:graph2q_a}, where every edge corresponds to a polarization identity involving the connected vertices. Horizontal and vertical edges denote measurements on separable states, while crossed edges denote measurements on entangled states. The graph $\mathcal{G}^{sep}_{2}$ illustrated in Fig.~\ref{fig:graph2q_b} only contains vertical and horizontal edges, and thus it is associated with polarization identities that involve separable states only. Remarkably, the graph $\mathcal{G}^{ent}_{2}$ can be transformed into the graph $\mathcal{G}^{sep}_{2}$ by the $10\leftrightarrow11$ vertex permutation. Furthermore, the action of this vertex permutation in the Hilbert space is implemented by the $\CNOT_{1,0}$ gate, defined by $\CNOT_{j,k}=|0\rangle_j\langle0|\otimes\mathbb{I}_k+|1\rangle_j\langle1|\otimes X_k$, where $j$-th qubit acts as a control and $k$-th qubit as a target. Thus, we define new bases $\{\tilde{\mathfrak{B}}_j\}$ as a rigid rotation ($\CNOT_{1,0}$) of the bases $\{\mathfrak{B}_j\}$, leading to   
 $\tilde{\mathfrak{B}}_0\equiv\mathfrak{B}_0$, $\tilde{\mathfrak{B}}_1\equiv\mathfrak{B}_1$, $\tilde{\mathfrak{B}}_2\equiv\mathfrak{B}_2$, and
$$\tilde{\mathfrak{B}}_3=\{|\pm,0\rangle,|\pm,1\rangle\},\,\tilde{\mathfrak{B}}_4=\{|\pm i,0\rangle,|\pm i,1\rangle\},$$
where the equivalence symbol ($\equiv$) means that these old and new bases are essentially the same up to permutation of elements and global phases, leading to fully equivalent measurements. Note that the old bases $\mathfrak{B}_0$, $\mathfrak{B}_1$, and $\mathfrak{B}_2$ are invariant under the action of $\CNOT_{1,0}$, whereas this non-local gate disentangles the old bases $\mathfrak{B}_3$ and $\mathfrak{B}_4$.

\begin{figure*}
        \centering
        \includegraphics[width=\linewidth]{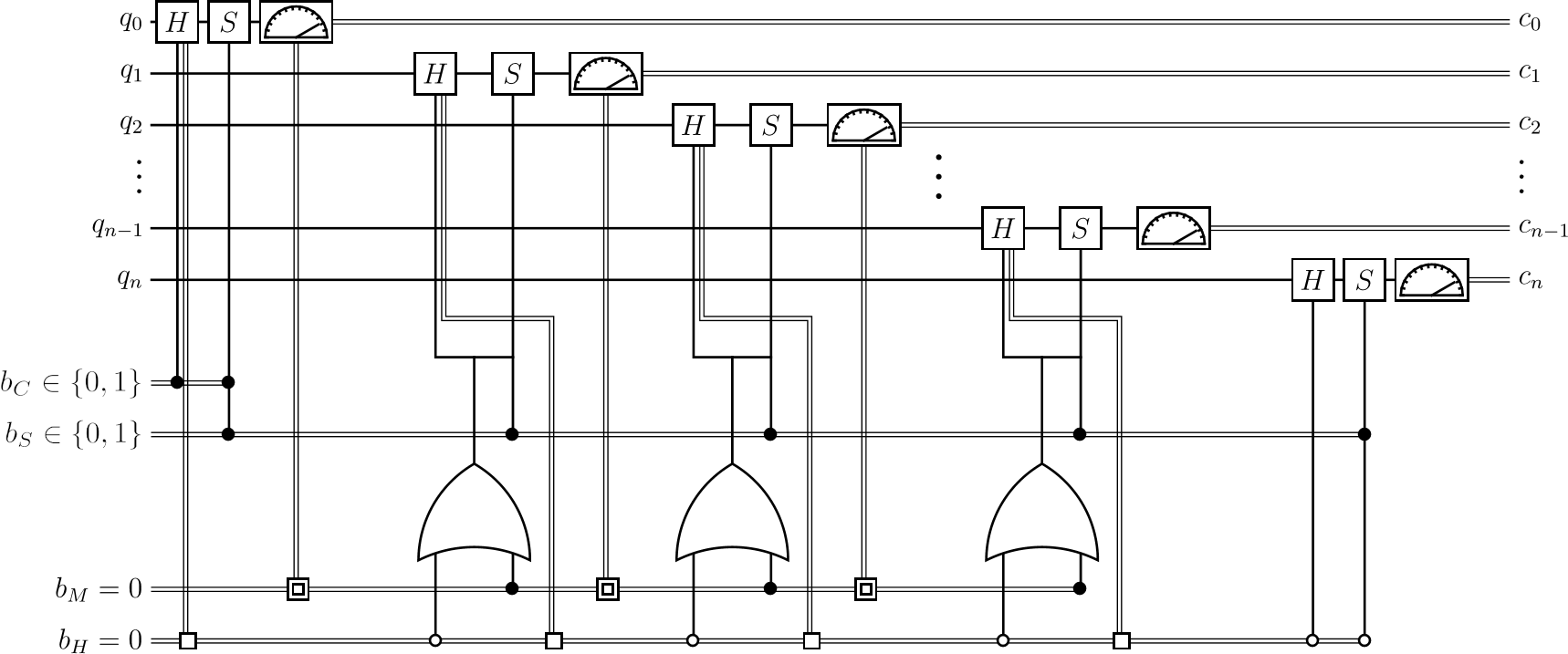}
        \caption{\justifying Classical-quantum circuit implementing bases $\tilde{\mathfrak{B}}_k$ with $k=1,\dots,4$. Horizontal single and double lines correspond to quantum and classical registers, respectively. Vertical single lines are always connected to quantum operations  applied conditionally on the values stored in classical registers, with open circle corresponding to conditioned on $0$, filled circle to conditioned on $1$ and we use a standard notation for the classical OR gate. 
        Vertical double line corresponds to operation applied to the classical register. Single square corresponds to \texttt{set} operation -- if the quantum gate has been applied, the corresponding classical register is set to 1, but not reverted. Double square corresponds to \texttt{overwrite} -- result of the corresponding measurement is written over whatever value the corresponding classical register had before.
        } 
        \label{fig:LOCC_Circuit}
    \end{figure*}

This approach can be extended to an arbitrary number of qubits. For a $n$-qubit system, the $2^n$-cycle graph, denoted as $\mathcal{G}^{ent}_n$, induces three fully separable $\mathfrak{B}_0,\mathfrak{B}_1,\mathfrak{B}_2$ and two genuinely entangled $\mathfrak{B}_3,\mathfrak{B}_4$ bases. Here, an edge $e_{\vb{j},\vb{j}'}$ with Hamming distance $d_H(\vb{j},\vb{j}') >1$, that is, binary strings $\vb{j},\vb{j}'$ such that the number of positions at which the corresponding characters or bits are different is greater than 1, implies an entangled state $\ket{\vb{j}}+\ket{\vb{j}'}$, or equivalently a crossed edge. This result is formally stated in the Lemma \ref{lem:ham_1_sep}. One of the main theoretical contributions of the present work consists in showing that a rigid rotation of the five bases $\mathfrak{B}_j$ produces five new bases $\tilde{\mathfrak{B}}_j$ composed of fully separable states for any number $n$ of qubits. This is summarized in the following;
\begin{restatable}[]{prop}{cnotprop}\label{prop:cnot_prop}
    Five basis of the form $\tilde{\mathfrak{B}}_i = S\mathfrak{B}_i$ with
    \begin{equation}\label{cnots}
    S = \CNOT_{n-1,n-2}\cdots \CNOT_{2,1}\CNOT_{1,0}
    \end{equation}
    consist exclusively in product states and can be performed by quantum circuits that require local operations and classical communication between the qubits.
\end{restatable}
The complete proof of this statement is given in the Appendix \ref{app:CNOT_separability}.

The chain of CNOT gates applied to the entangled bases $\mathfrak{B}_3$ and $\mathfrak{B}_4$ reduces the Hamming distances to 1 for each state in the bases, while the bases $\mathfrak{B}_0$, $\mathfrak{B}_1$ and $\mathfrak{B}_2$ remain disentangled after applying the chain of the CNOT gates. Thus, we arrive at a connected cycle in the hypercubic graph. This is depicted in Fig.~\ref{fig:hypercube_3q} for the 3-qubit case. In other words, the graph $\mathcal{G}^{sep}$ associated with the 5 disentangled bases can be obtained from the graph $\mathcal{G}^{ent}$ by applying the chain of CNOT gates to all states $\ket{\vb{j}}$ associated with the vertices $v_{\vb{j}}$. That is, we obtain a $2^n$-cyclic graph, where the absence of crossed edges implies that none of the vectors belonging to the $n$-qubit bases is entangled.

Fig.~\ref{fig:LOCC_Circuit} displays a hybrid classical-quantum circuit that creates disentangled bases $\tilde{\mathfrak{B}}_i$ with $i=1,2,3,4$. The circuit consists of Hadamard gates $H=|+\rangle\langle0|+|-\rangle\langle1|$ and phase gates $S=|0\rangle\langle0|+i|1\rangle\langle1|$ that act locally on the collection of qubits $q_k$ with $k=0,\dots,n$, followed by a measurement of the computational basis. The circuit also considers 4 bits $b_C$, $b_S$, $b_M$, $b_H$. The bit $b_M$ is initialized as $b_M=0$ and progressively overwritten with the measurement results on the canonical basis. The bit $b_H$ is initialized as $b_H=0$ and stores whether or not a Hadamard gate has been applied. If the quantum gate has been applied, bit $b_H$ is set to 1 but not reverted. The initial settings of bits $b_C$, $b_S$ control the unentangled basis to be implemented. Thereafter, the action of the quantum gates $H$ and $S$ is controlled by the value of bits $b_C$, $b_S$. The classic OR gate is applied to bits $b_M$ and $b_H$, and its result is used to control the action of quantum gates $H$ and $S$. The hybrid classical-quantum circuit in Fig.~\ref{fig:LOCC_Circuit} is based on an algorithm that implements an alternative analytic construction of the disentangled bases. See Appendix \ref{sec:LOCC_prot} for a detailed explanation of this method.

\section{\texorpdfstring{$2n+1$}{2n + 1} Tensor product bases}
\label{sec:2n+1_basis}

In this section, we derive a set of $2n+1$ informationally complete and fully separable bases for $n$-qubit pure states, arising from the polarization identity (\ref{pol_identity}). As an improvement to the five bases introduced in Section \ref{sec:5bases}, these bases do not require classical communication between the qubits to be implemented. However, the price to pay is that the number of bases is not fixed but scales linearly as $2n+1$ with the number of qubits $n$. However, we demonstrate that this number of bases is minimal, that is, fewer than $2n+1$ separable bases arising from the polarization identity are not informationally complete for pure states.

To derive these results, we first note that all the states belonging to the measurement bases have the form
\begin{equation}\label{states_pol}
|\psi\rangle=|\vb{j}\rangle+(i)^\ell|\vb{j}'\rangle
\end{equation}
according to the second line in (\ref{pol_identity}).
% The full separablility of these states imposes the following restrictions:
% \begin{equation}\label{purities}
% (j_i\oplus j'_i)(j_k\oplus j'_k)=0, \text{ for all }i\not=k.
% \end{equation}
% To show this, one has to impose every single qubit reduction of $|\psi\rangle$ is a pure quantum state. Thus, under the assumption that $\vb{j} \ne \vb{j}'$, eq.~(\ref{purities}) implies that the Hamming distance ($d_h$) between $\vb{j}$ and $\vb{j}'$ is 1, with
% $d_h(\vb{j}, \vb{j}') = \sum_{k=1}^n \left( j_k \oplus j'_k \right)$.
The separability of these states is captured directly by the Hamming distance between the underlying binary numbers.
\begin{lem}\label{lem:ham_1_sep}
    A pure quantum state of the form $$\ket{\psi} = a\ket{\vb{j}} + b\ket{\vb{j'}}$$
    is separable if and only if $\vb{j}$ and $\vb{j}'$ have Hamming distance one, i.e., $d_H(\vb{j},\vb{j}') \equiv \sum_{i=0}^n \abs{j_i - j'_i} \leq 1$.
\end{lem}
\begin{proof}
    % Assume that $d_H(\vb{j},\vb{j}') = m$ \dg{I think this phrase can be removed}. 
    We take $d_H(\vb{j},\vb{j}')>0$, as the case of $\vb{j} = \vb{j}'$ is trivial. Let us define the subset $I = \qty{i: j_i = j'_i}$ of identical indices, that is, indices where both binary strings $\vb{j}$ and $\vb{j}'$ have the same bit-value, and its complement $\overline{I}$. Thus, the state $\ket{\psi}$ can be written as a product between subspaces corresponding to both subsets, $\ket{\psi} = \ket{\psi}_I \otimes \ket{\psi}_{\overline{I}}$, where $I$ could be the empty set, for example, for a GHZ-like state. The state $\ket{\psi}_I = \bigotimes_{i\in I} \ket{j_i}$ is fully separable. By definition, we have $\ip{\vb{j}}{\vb{j}'} = \delta_{\vb{jj}'}$. Since an arbitrary phase of $b$ can be canceled by local operation $e^{i\lambda Z}$ acting on any of the qubits, we may take $a, b \in\mathbb{R}^+$ without loss of generality. This ensures that $\ket{\psi}_{\overline{I}}$ is written in its Schmidt form for any bipartition of $\overline{I}$. Thus, $\ket{\psi}_{\overline{I}}$ is fully separable if and only if $\abs{\overline{I}} = 1$, which is equivalent to $d_H(\vb{j},\vb{j}') = 1$.
\end{proof}
% \jcz{I rewrote the fragment (now commented) to accommodate the more organised form I suggested before}

The above lemma implies that the graph induced by these measurements is the hypercube~\cite{harary1969graph}. Here, each set of edges composed of $2^n$ parallel edges in the hypercube defines an orthonormal basis. Therefore, we have a set made up of $2n$ fully separable measurement bases. We also include the computational basis within our set, as this allows us to mitigate error propagation in the estimation process, as we show in the Appendix \ref{app:selection_procedure}. Thus, our second estimation method requires $2n+1$ bases.

Additionally, it is straightforward to show that there is a unique way to define sets of orthonormal bases for these vectors, and any removal of a basis implies a disconnected graph $\mathcal{G}^{sep}$, since each basis represents $2^n$ parallel vertices in the $n$-dimensional hypercube; see Fig. \ref{fig:hypercube_3q}. Thus, the informationally complete set of fully separable bases arising from the polarization identity is unique.

For the case of a 2-qubit system, these bases coincide with the 5 disentangled bases defined in Section \ref{sec:5bases}. For 3-qubits, the 7 bases are explicitly given in the Appendix \ref{SM:7bases}.  

For an $n$-qubit system, there is a simple way to implement the $2n+1$ bases in a quantum computer, which involves a single-qubit gate applied to a single qubit per circuit. The first set of $n$ circuits requires a single-qubit Hadamard gate $H$ applied to the $k$-th qubit, $k=1,\dots,n$. The second set of $n$ circuits applies the single-qubit gate $\widetilde{H}=SH$ instead of $H$ on the $k$-th qubit, where
\begin{equation}\label{eq:operations_definition}
        H = \frac{1}{\sqrt{2}}\mqty(1 & 1 \\ 1 & -1)\mbox{ and } S = \mqty(1 & 0 \\ 0 & i).
    \end{equation}
The quantum circuits for generating separable the measurement bases in the 3-qubit case are depicted in Fig.~\ref{fig:bases_3qubits}.

% \dg{\textbf{Tomás}. Please add this figure} \vdga{Done, I think} . 
\begin{figure}
                \centering
                \begin{subfigure}{0.5\textwidth}
                    \centering                    \includegraphics[width=0.22\textwidth]{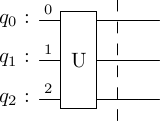}
                    \caption{Computational Basis.}
                    \label{CB}
                \end{subfigure}\\
                \begin{subfigure}{0.15\textwidth}
                    \includegraphics[width=\textwidth]{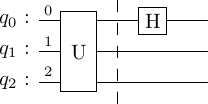}
                    \caption{$H_0.$}
                    \label{H_0}
                \end{subfigure}
                \begin{subfigure}{0.15\textwidth}
                    \includegraphics[width=\textwidth]{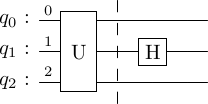}
                    \caption{$H_1.$}
                    \label{H_1}
                \end{subfigure}
                \begin{subfigure}{0.15\textwidth}
                    \includegraphics[width=\textwidth]{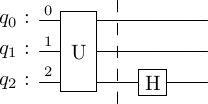}
                    \caption{$H_2.$}
                    \label{H_2}
                \end{subfigure}\\
                \begin{subfigure}{0.15\textwidth}
                    \includegraphics[width=\textwidth]{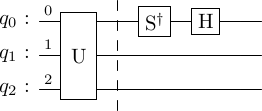}
                    \caption{$SH_0.$}
                    \label{SH_0}
                \end{subfigure}                
                \begin{subfigure}{0.15\textwidth}
                    \includegraphics[width=\textwidth]{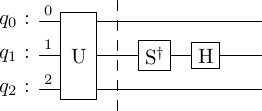}
                    \caption{$SH_1.$}
                    \label{SH_1}
                \end{subfigure}                
                \begin{subfigure}{0.15\textwidth}
                    \includegraphics[width=\textwidth]{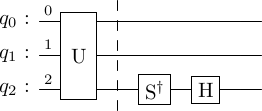}
                    \caption{$SH_2.$}
                    \label{SH_2}
                \end{subfigure}
                 
                \caption{\justifying Quantum circuits for the generation of the 7  separable measurement basis for a $3$-qubit system.}
                \label{fig:bases_3qubits}
\end{figure}
\begin{obs}
The weights of a vertex in $\mathcal{G}$, that is, ${w_{\vb{j}}=|\langle\psi|\vb{j}\rangle|^2}$, depend on the quantum state $|\psi\rangle$ to be estimated. A vanishing weight implies the removal of the associated vertex from the graph, as it does not play any role in the estimation process. 
\end{obs}

The two estimation methods presented here share a common feature with the 5-bases method, that they fail when two non-consecutive coefficients, considering the lexicographical order of indices, of the state $|\psi\rangle$ vanish, or equivalently, when the number of vanishing weights produces a disconnected graph. In such a case, the methods provide an infinite number of solutions. The states affected by this characteristic are a set of null measure, and consequently, their occurrence is unlikely.

% \begin{obs}
% \jcz{What was meant to be in this observation?}

% \end{obs}

The hypercubic graph contains $2^{-n}\prod_{i=1}^n (2i)^{\binom{n}{i}}$ spanning trees~\cite{Stanley_Fomin_1999}, each of them admitting a number of recursive methods to estimate a pure quantum state. In an error-free implementation, all these methods lead to the same estimate. However, when considering realistic errors, each path produces a different estimate of the unknown pure state. As a consequence, there is a large freedom to optimize with respect to errors. Appendix \ref{app:errors} contains a detailed explanation of error propagation, and Appendix \ref{app:analityc_sol} the analytic solution to the system of equations that estimate the state, regardless of the chosen subgraph.

As a final comment of the section, let us note that the statistics collected from the five, or the $2n+1$ bases, is informative enough to determine whether the prepared state is rank 1 or not. In short, polarization identity (\ref{pol_identity}) connecting vertices $v_{\vb{j}}$ and $v_{\vb{j}}'$ allows us to determine whether the $\vb{j}$-th and $\vb{j}'$-th entries of the density matrix are linearly dependent. In this way, a connected graph allows us to certify whether the prepared state is pure or not. See details in Appendix \ref{SM:purity}.

\section{Demonstration on IBM quantum processing units}\label{sec:IBM}

    To evaluate the performance of our estimation methods, we performed simulations on IBM quantum processors ibm{\textunderscore}kyoto, ibm{\textunderscore}kyiv, and ibm{\textunderscore}sherbrooke, each one of them is an IBM Quantum Eagle r3 processor. This has 127 qubits connected through a heavy-hexagonal architecture~\cite{abughanem2025ibm}. As a consequence, every qubit is physically connected to two or three neighboring qubits. 

    \medskip
    \noindent \textit{State preparation:} To test both estimation methods, we generate genuinely entangled $n$-qubit graph quantum states. These are created using constant depth-three circuits comprising $2n-1$ bipartite entangling gates for an $n$-qubit system. Any smaller number of such gates fails to generate genuine entanglement in a depth-three circuit. A similar constant-depth architecture was previously considered in~\cite{cabello2011optimal}. 
    % Furthermore, although it is possible to prepare any pure state in Eagle r3 quantum processors, this is achieved by applying a universal set of native gates, composed in this case by the ECR gates, $\mathcal{I}$,$R_z$, $SX$ and $X$, which implies that any operation performed to prepare the state must be represented as a series of operations of this universal set of native gates. 
    Our state preparation procedure uses $\mathbb{I}$,$\sqrt{X}$ and ECR gates, which are native gates in the Eagle r3 QPU. The prepared states are given by:

    %The quantum circuit associated to the preparation of this state consists of three layers acting on an input state $\ket{0}$. 
    % \jcz{Is this correct?} \vdga{This is not the best way to describe the state I have been thinking in a generalized way to write down the state, the amplitudes are of this form but I have not find a pattern for the phases} \jcz{Looking forward to it!}. 
    %The first layer applies a single-qubit $\sqrt{X}$ gate on each qubit. The second layer consists of applying two-qubit ECR gates with the qubit $q_j$ as control qubit to $q_{j+1}$ with $j$ a even number. In case $n$ odd, an additional $\sqrt{X}$ gate is applied to the last qubit. The third layer also consists of ECR gates but starting from qubit $q_{j}$  to the  $q_{j+1}$ with $j$ odd. In addition, a $\sqrt{X}$ gate is applied to $q_0$. It can be written explicitly as
    % \begin{equation}
    % \qty(S_X\otimes \text{ECR}^{\otimes \lfloor \frac{n-1}{2}\rfloor})
    % \qty(\text{ECR}^{\otimes \lfloor \frac{n}{2}\rfloor}\otimes S_X^{\otimes n - 2\lfloor \frac{n}{2}\rfloor})
    % S_X^{\otimes n}    
    % \end{equation}}\jcz{please check this} 

    %  \vdga{ Eqs. \eqref{eq:state_n_even} and \eqref{eq:state_n_odd} fit to state preparation now.
    \begin{widetext}
    \begin{equation}
         \ket{\psi_2}= \left(\sqrt{X}\otimes \mathbb{I}_2\right)\left(ECR_{0,1}\right)\nonumber(\sqrt{X}\otimes \sqrt{X})\ket{00}
    \end{equation}
    for two qubit systems, 
    \begin{equation}\label{eq:state_n_even}
        \ket{\psi_n}= \left(\sqrt{X}\otimes\ \bigotimes_{k=0}^{\lfloor \frac{n}{2}\rfloor-2}ECR_{2k+1,2k+2}\otimes \mathbb{I}_2\right)\left(\bigotimes_{j=0}^{\lfloor \frac{n}{2}\rfloor-1}ECR_{2j,2j+1}\right)\nonumber(\sqrt{X}^{\otimes n})\ket{0}^{\otimes n}
    \end{equation}
    for even $n>2$, and
    % \jcz{Second brackets act only on $n-2$ systems, unless we implicitly assume there is implicit $\mathbb{I}_4$ in the picture. Shouldn't it be $\bigotimes^{n/2}$ simply?} \vdga{I think the equation covers all subspaces, for example for $n=4$ $j=0,1$, $k=0$ resulting in 
    % \begin{equation}
    %     \left(\sqrt{X}\otimes ECR_{1,2}\otimes \mathbb{I}\right)\left(ECR_{0,1}\otimes ECR_{2,3}\right)\left(\sqrt{X}^{\otimes n}\right)\ket{0000}
    % \end{equation} } 
    \begin{equation}        
    \label{eq:state_n_odd}
        \ket{\psi_n}= \left(\sqrt{X} \otimes\bigotimes_{k=0}^{\lfloor \frac{n}{2}\rfloor-1}ECR_{2k+1,2k+2}\right)\left(\bigotimes_{j=0}^{\lfloor \frac{n}{2}\rfloor-1}ECR_{2j,2j+1} \otimes \sqrt{X}\right)\nonumber(\sqrt{X}^{\otimes n})\ket{0}^{\otimes n},
    \end{equation}
    for $n>1$ odd, where $j\in \{0,n-1\}$.
    % \jcz{Something doesn't add up here... red parts seem to need a double check}\vdga{I think it is good now}\jcz{There is still one problem -- for $n=3$ we have $\lfloor3/2\rfloor - 1 = 0$. Shouldn't this be ceiling, not floor?}\vdga{It is actually floor, for $n=3$ we need to apply only one $ECR$ gate on each layer but the first one the we have the in the second layer $ECR_{0,1}\otimes \sqrt{X}$ and $\sqrt{X}\otimes ECR_{1,2}$ for the third one. Look at Fig. \ref{fig:3qb_system} P.D I've just noted that there is an error in  Fig \ref{fig:4qb_system} a $\sqrt{X}$ in missing on $q_3$ in the first layer.}
    % \jcz{Even more reason to streamline numbering, especially with respect to CNOTs in the beginning...}\dg{\textbf{Jakub}. I don't see any problem with these expressions. In each set of parallel gates there are $n-1$ non-local gates when $n$ is even. For $n$ odd the first set is composed of $n-1$ and the second set of $n-1$ non-local gates. Please check this explicitly from Figure \ref{fig:bases_3qubits}.}
    \end{widetext}
    
    % \jcz{Self-comment after seeing the plot in Fig.\ref{fig:states_preparation} I know the above expression is incorrect -- please try to adjust the statement so that it reflects the actual preparation procedure as depicted in Fig.~\ref{fig:states_preparation}} 
    % This generates a state with genuine entanglement such that the square modulo of the probability amplitudes is equal for each element of the state when it is represented in the computational basis. 
     
    % Moreover, due to the limited connectivity of the qubits in the processors, mapping the circuit on the hardware is not trivial and, despite IBM provides methods for help with the qubit selection, the qubits obtained by these not always have the best performance.

    Figs.~\ref{fig:3qb_system} and \ref{fig:4qb_system} show the quantum circuits for $3$ and $4$ qubits, respectively. Note that the generated states $\ket{\psi_n}$ have constant amplitudes $\abs{\ip{\vb{j}}{\psi_n}}^2$ for all $\vb{j} = 0,\hdots,2^n-1$, for all $n\in\mathbb{N}$.
    
    Due to the restricted qubit connectivity of the Eagle r3 QPU, finding an efficient circuit implementation in practice is a non-trivial problem. For this reason, we developed a new optimization algorithm that selects the best subset of physical qubits; see Appendix \ref{app:selection_procedure}. This dedicated algorithm outperforms Qiskit’s built-in transpiler at optimization level 3.

\begin{figure}
    \centering
        \begin{subfigure}{0.23\textwidth}
            \includegraphics[width=\textwidth]{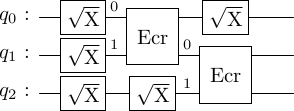}
            \caption{$3$-qubit system.}
            \label{fig:3qb_system}
        \end{subfigure}
        \hspace{0.2cm}
         \begin{subfigure}{0.23\textwidth}
            \includegraphics[width=\textwidth]{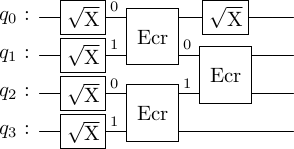}
            \caption{$4$-qubit system.}
            \label{fig:4qb_system}
        \end{subfigure}
    \caption{\justifying Efficient preparation of a genuinely entangled state by using native gates from IBM Eagle r3 QPU.
    %\dg{All figures have to show the same label for qubits at the beginning. Please make this compatible for all the figures.}\vdga{I'll fix it as soon as possible.}
    }
    \label{fig:states_preparation}
\end{figure}       

    \medskip
    \noindent\textit{Estimation using $2n+1$ tensor product bases:} We have implemented the estimation of a $n$-qubit graph state $\ket{\psi_n}$ with $n=3,\dots,12$. The experiment is carried out with a total of $N_T=N(2n+1)$ independently and identically prepared copies of the graph state, where $N$ is the number of copies measured with each basis, or the number of shots per measurement basis in the quantum computing language. We considered $N=2\times10^4$ shots for $n=3,\dots,9$, $N=4\times10^4$ shots for $n=10,11$, and $N=8\times10^4$ shots for $n=12$. The fidelity $F=|\langle\tilde\psi_n\ket{\psi_n}|^2$ between the estimated state $\ket*{\tilde{\psi}_n}$ and the target state $\ket{\psi_n}$ is shown in Fig.~\ref{fig:results_2n} as a function of the number $n$ of qubits for three IBM quantum computers. Each curve shows a decrease in fidelity with an increasing number of qubits, with the single exception of the Kyev quantum computer at $n=11$, where there is an increase in fidelity. This general behavior is expected since the total number of real parameters to be estimated increases as $2^{n+1}-2$ while the number $N$ of shots per measurement basis increases in the experiment by a factor of 4 at most. Therefore, we have an exponential increase in the number of parameters but a linear increase in the total number of shots. However, fidelity can be increased to a value close to 1, provided a sufficiently large total number of shots can be generated. This is shown in Fig. \ref{fig:results_12_qb_ensemble_size}, where fidelity $F$ as a function of the total number of shots $N_T$ is illustrated for the case of $n=12$. As can be seen in this figure, a fidelity of 0.937 is achieved.

    %For a $12$-qubit quantum system, we achieved  fidelity of $f\approx 0.801$ in IBM Sherbrooke, $f\approx 0.794$ in IBM Kyoto and $f\approx 0.789$ in IBM Kyiv. 
    \begin{figure}[htb]
        \centering        
        \includegraphics[width=.9\columnwidth]{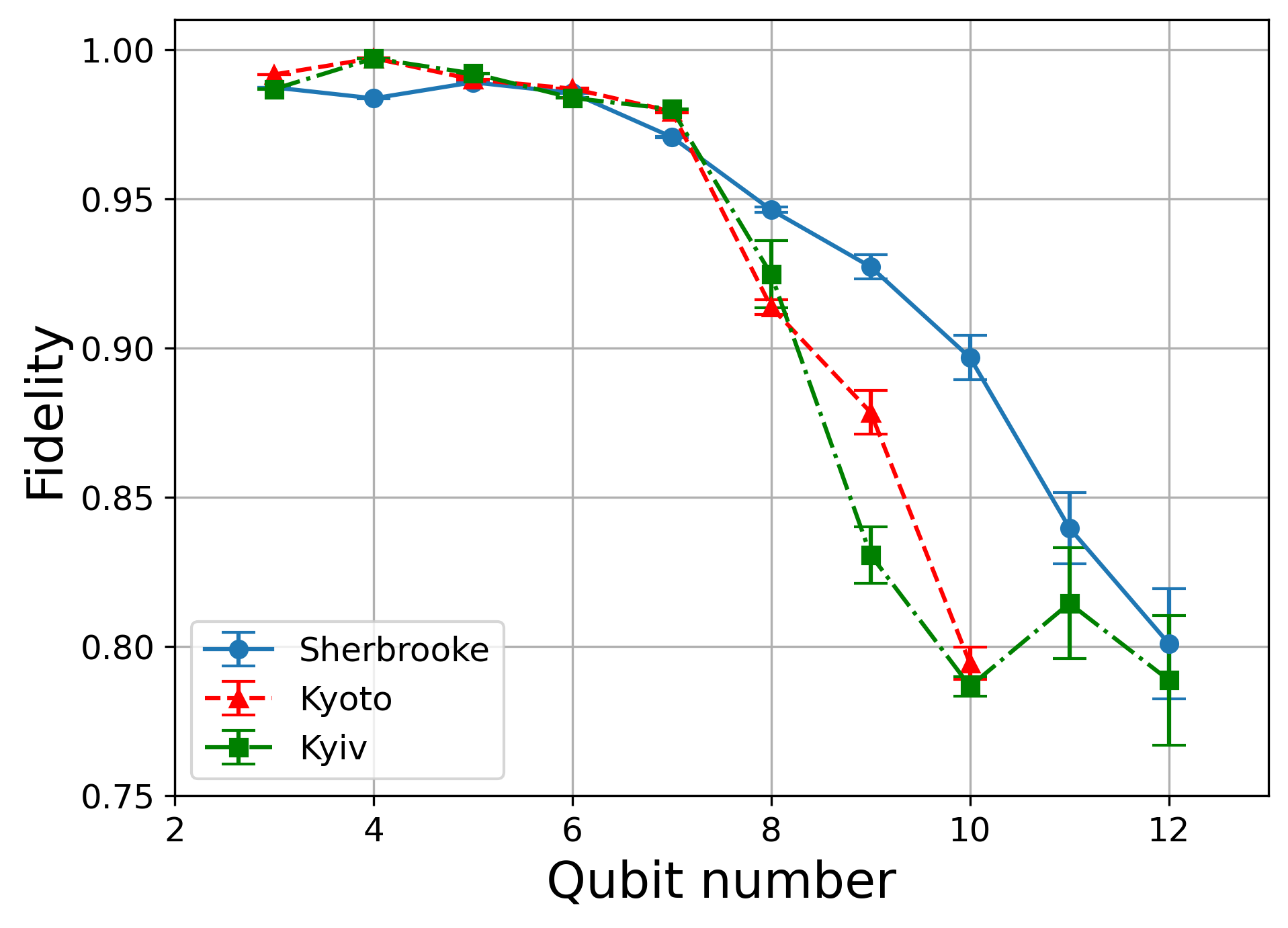}
        \caption{\justifying (Color online) Fidelity reconstruction vs. number of qubits for three free accessible IBM quantum computers. For 11 qubits there is an apparently unexpected behaviour of Kyiv experiments perhaps due to errors, as this figure involves single multishot experiments with fixed qubit subsets and hardware callibration, and not an average over a large number of realizations. As such, results may deviate from expected monotonic behavior due to setup specifics, which could not have been accounted for. Error bars were estimated by assuming a multinomial statistical distribution of data.
        % \jcz{Since later on we invoke statistical error, it would be reasonable to include the error bars. Can anybody do it?} \vdga{I'll try to work on error bars for all this figures.}
        %\jcz{Error bars are awesome! But it feels like 11-qubit freak case for Kyiv needs an explanation -- even if it was connected to a necessity of involving long-range connections or a faulty qubit... but the monotonicity clearly persists in spite of the error bars.}\dg{This is a single implementation and error bars consider statistical effects only. There is no reason to have a strict monotonical behavior here, despite it is highly probable. Do you agree the phrase in olive?}\jcz{I have a small problem with it -- they are single experiments in the sense that they were done with a single configuration, and thus do not involve different selections of qubits or callibration settings, correct? If this is so, my comment still applies, even more so. On the other hand, I don't think we want to give off a feeling as if we are doing it \textbf{single-shot} -- I added a small fragment in blue}\dg{Victor: Is it true that you considered multinomial statistical distribution to estimate errors? If not, please improve the last phrase (olive color). \vdga{Yes, I considered multinomial statistical distribution to estimate errors from direct measurement results, and error the error propagation given by the procedure sowed in Appendix \ref{app:errors}.}}
        }
        \label{fig:results_2n}
    \end{figure} 
    \begin{figure}[htb]
        \centering
        \includegraphics[width=.9\columnwidth]{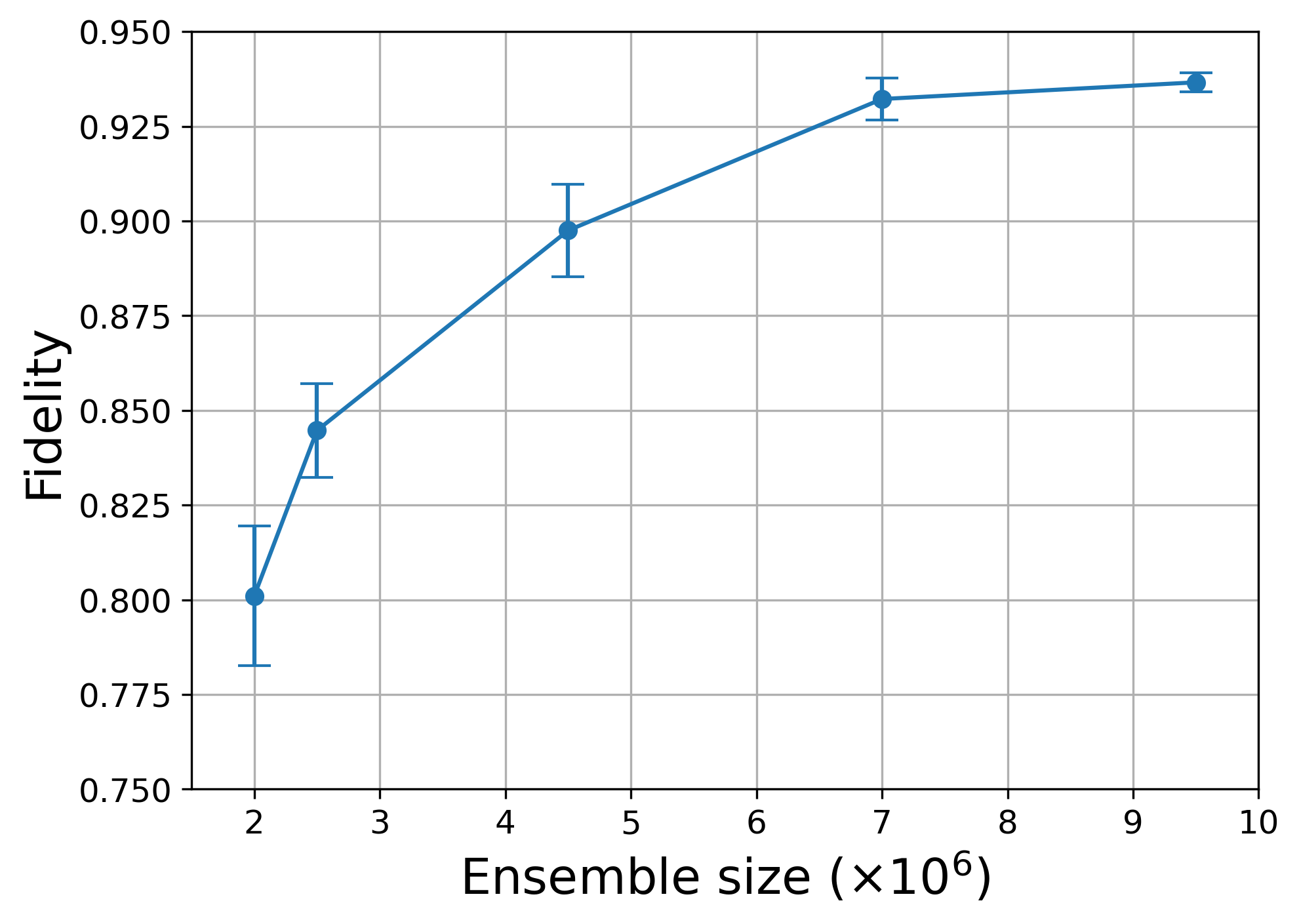}
        \caption{\justifying Fidelity vs. total ensemble size obtained with IBM Sherbrooke for the case of 12 qubits. 
        % \dg{Victor: please consider the vertical axis in the range $[0.75-0.95]$. In the horizontal axis, is $(\times 10^6)$ better than $(10^6)$?}\vdga{Done!}
        }        \label{fig:results_12_qb_ensemble_size}
    \end{figure}

    \medskip
        
    \noindent\textit{Estimation using five disentangled bases:} Measurement bases $\tilde{\mathfrak{B}}_k$ with $k=0,\dots,4$ only contain fully separable states. Unlike the bases $\tilde{\mathfrak{B}}_0$, $\tilde{\mathfrak{B}}_1$, and $\tilde{\mathfrak{B}}_2$, which are fully separable, the measurement bases $\tilde{\mathfrak{B}}_3$ and $\tilde{\mathfrak{B}}_4$ are generated using classical communication as a resource. Figure~\ref{fig:LOCC_Circuit} shows the quantum circuits required to generate the five disentangled measurement bases.
    
    Estimation of genuinely entangled graph states using disentangled bases is carried out with $N=10^5$ shots per measurement basis, for $n=3$ to $n=7$ qubits. As before, we calculate the fidelity $F=|\langle\tilde\psi_n\ket{\psi_n}|^2$ between the estimated state $|\tilde{\psi}_n\rangle$ and the ideal state $\ket{\psi}$. Fig. \ref{fig:protocols_comparison} compares the performance of $2n+1$ and 5 disentangled bases. The first row considers $2\times 10^4$ shots per basis. In this case, the $2n+1$ bases method outperforms the five bases in almost all the cases studied. However, when considering the same number $N_T=10^5$ of total shots for each estimation method, the five disentangled bases exhibit a clear advantage over the $2n+1$ bases method, which is shown in the second row of Fig. \ref{fig:protocols_comparison}. In particular, the $2n+1$ bases method achieves a fidelity close to 0.8 for $n=6$ and 7 qubits, while the method of the five disentangled bases achieves a fidelity close to 0.95, which represents a very large gain. This is true for the results obtained on all three IBM QPUs.

    The above results indicate that for a fixed total number of shots, it is advisable to prefer the method of the five disentangled bases instead of the $2n+1$ bases method. This is because the method of the five disentangled bases allocates more shots per measurement basis, thus reducing the measurement errors per basis and consequently increasing the fidelity of the estimated state. As long as finite statistics is the only source of noise, the method of the five disentangled bases should deliver a higher estimation fidelity than the $2n+1$ bases method. However, these methods rely on different resources, classical communication and single-qubit gates, respectively, and therefore it is not clear which achieves better performance for a large number of qubits.

\begin{figure*}
    \centering
    \includegraphics[width=\textwidth]{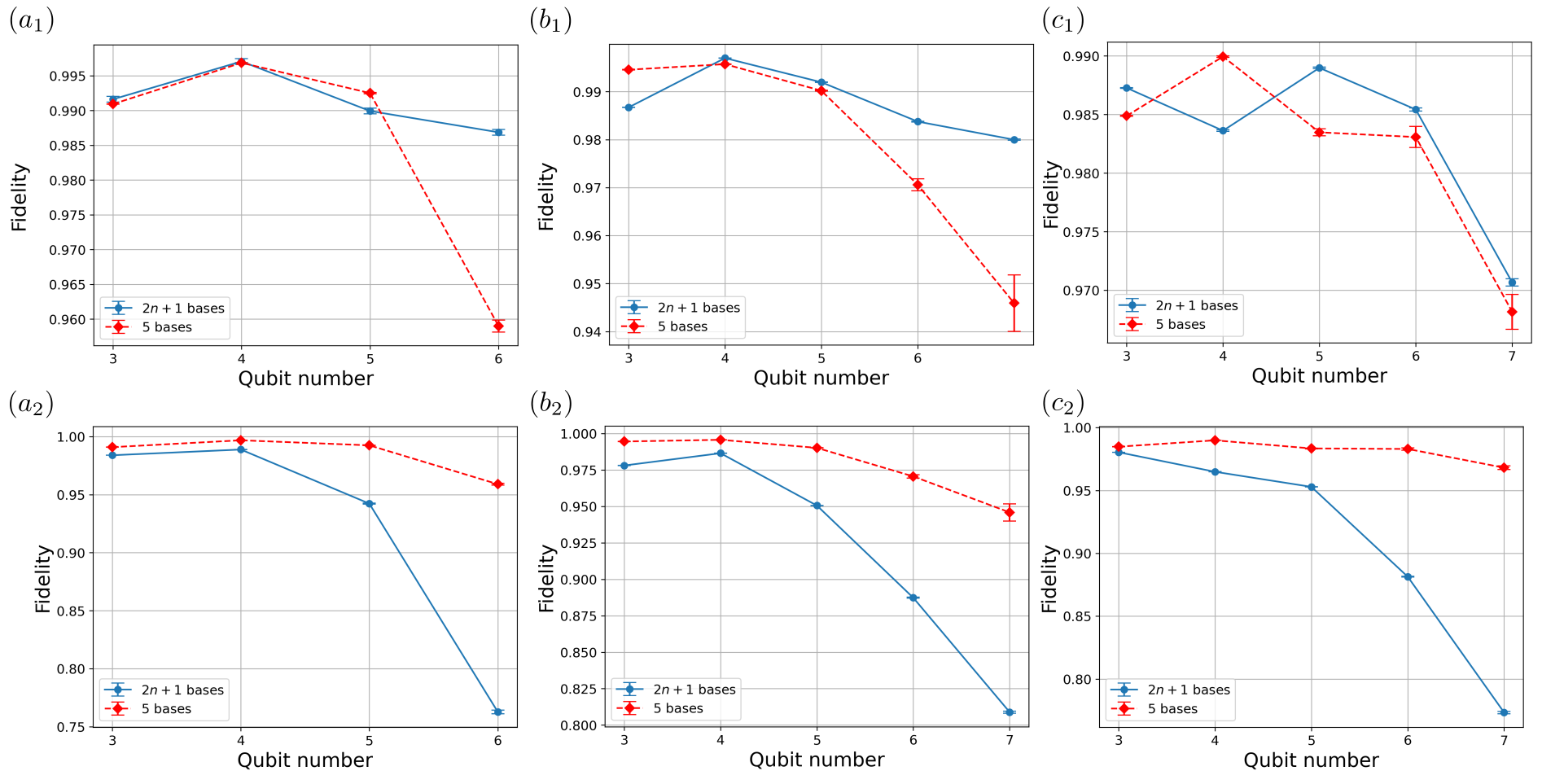
    }
    \caption{\justifying (Color online) Comparison of fidelity achieved between the $2n+1$ (blue lines) and $5$ (red lines) bases methods vs. the qubit number of the system. On the first row of plots, we consider $2\times10^4$ shots per basis, whereas the second row considers a fixed number of total shots $10^5$, equally distributed among all the bases in each case. Each column represents experiments on a real quantum processor, namely, IBM Kyoto, Kyiv and Sherbrooke for subfigures $(a_1)$-$(a_2)$,   $(b_1)$-$(b_2)$ and $(c_1)$-$(c_2)$, respectively.}
    \label{fig:protocols_comparison}
\end{figure*}            
            % \vdga{
        % From here, we can conclude that if the ensemble size is identical, the 5 bases method is superior, given that the experiments are distributed in a small set of measurements, getting a better estimation of the actual probability distribution. Then, if we deal with limitations in the number of experiments, it is recommended to use the $5$ bases method. However, if we can perform more experiments, it is better to use the $2n+1$ bases. Moreover, the $2n+1$ bases method is easier to implement and has a lower error propagation rate. It is important to note that the measurement process required in the LOCC protocol of the $5$ bases method can introduce some extra error in the experiments.

\section{Experiment with remote entangled trapped ions}
\label{sec:Oxford}

In order to validate our method, we analyze the data presented in~\cite{main2025multipartite}, where two remote ion-trap small quantum processors are entangled. In this experiment, two ion traps---denoted \emph{Alice} and \emph{Bob}---are separated by approximately \(2\,\mathrm{m}\), each loaded with one \({}^{88}\mathrm{Sr}^+\) ion and one \({}^{43}\mathrm{Ca}^+\) ion (see Fig.~\ref{fig:ions}). Qubits can be encoded on the electronic states of each of the atomic ions. Remote entanglement between the two separated \({}^{88}\mathrm{Sr}^+\) ions is established via single-photon emission and Bell-state analysis of the emitted photons~\cite{stephenson2020high}, producing the Bell state $|\Psi^+\rangle$ between the two strontium ions.  

\begin{figure}
    \centering
    \includegraphics[width=\columnwidth]{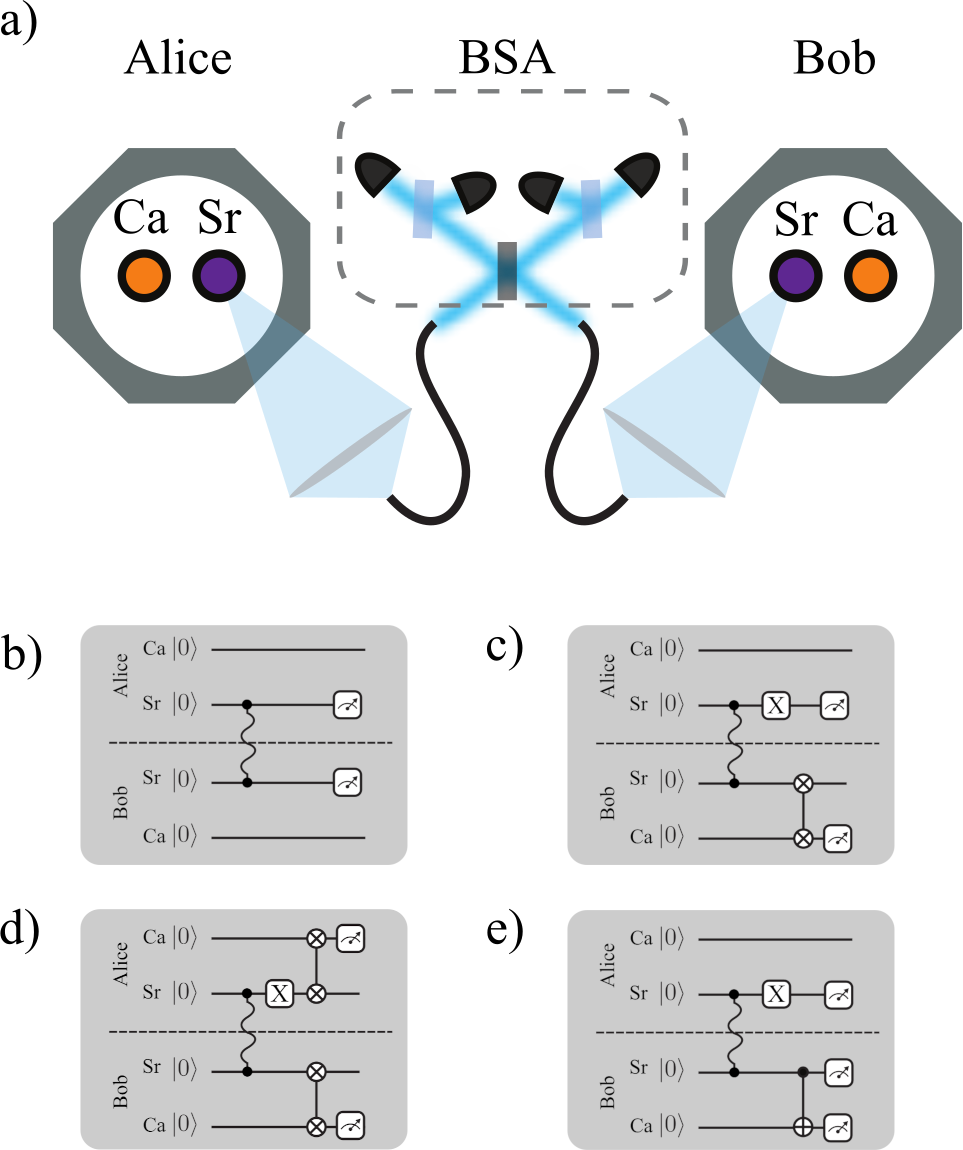}
    \caption{\justifying Experimental configuration and corresponding quantum circuits for generating remote entangled states between trapped ions. a) Experimental setup for generating remote entangled states between trapped ions. Two small processors (Alice and Bob), each containing a single $^{88}$Sr$^+$ ion, are entangled via interference of the photons they emit on a Bell-state analyser (BSA), which projects the ions into a Bell state. Local single- and two-qubit gates are then used to map entanglement between the Sr ions onto the co-trapped $^{43}$Ca$^+$ ions. b)–e) Quantum circuits implemented to generate the two- and three-qubit states listed in Table~\ref{tab:summary}. The wavy line denotes remote entanglement, while the measurement symbol indicates a single-qubit measurement in an arbitrary basis.
    % \jcz{Do you think the figure can be rearranged to have a) on top and b-e) as a 2x2 grid blow? Currently it feels a bit too large, no?}\jcz{Tentative alternative in place!}
    }
    \label{fig:ions}
\end{figure}    

The local single- and two-qubit gates are then applied to create bipartite entangled states of \(\mathrm{Sr}_A\!-\!\mathrm{Ca}_B\), \(\mathrm{Ca}_A\!-\!\mathrm{Ca}_B\), and the GHZ state of three-qubits \(\mathrm{Sr}_A\!-\!\mathrm{Sr}_B\!-\!\mathrm{Ca}_B\), where the subscripts indicate the trap (\emph{Alice} or \emph{Bob}) in which each ion resides. The quantum circuits used to generate these states are shown in Figs.~\ref{fig:ions}b)-e). A detailed description of the experiment, including different sources of errors in the preparation of the states and measurements, can be found in ref.~\cite{main2025multipartite}.  

The fidelity of the generated states with respect to the ideal Bell or GHZ states can be estimated using partial- or full state tomography. Table \ref{tab:summary} reports the fidelity obtained using three different methods: Parity measurements (PM)\cite{monz201114}, standard Pauli tomography (PT) with maximum-likelihood estimation~\cite{vrehavcek2007diluted, stephenson2020high}, and the new $2n+1$-bases method introduced in this work. In the case of PT, each qubit is measured in each of the three Pauli bases, resulting in \(3^n\) distinct measurement settings for an \(n\)-qubit state. From the reconstructed density matrices using PT, the purities of the states were also calculated (Table \ref{tab:summary}). Using the $2n+1$-basis method presented in this work, we select a subset of the Pauli basis measurements described above to implement the optimized protocol and reconstruct a pure state compatible with the measurements. 
The subset of Pauli basis measurements used for the method introduced in this work are those such that all qubits except one are measured on the computational basis (The basis given by $\sigma_z$), and the remaining qubit is measured either in the Pauli $\sigma_x$ or $\sigma_y$ basis, these measurements are equivalent to the described in Section \ref{sec:2n+1_basis} up to a global phase.
% [\gam{\textbf{Victor}: can you please add the details about what basis you constructed from the bases that we measured? Thanks}]
The fidelity values obtained with this set are given in Table \ref{tab:summary}. The density matrices reconstructed using full tomography and using our optimized method are shown in the Appendix~\ref{SM:IonTrap}. Furthermore, Table~\ref{tab:summary} shows the overlap (fidelity) between the reconstructed density matrices using PT and the $2n+1$ bases tomography. Notably, even given the limited purities of the experimentally generated states, the fidelity estimates from the optimized method show good agreement with those from the standard, non-optimized tomography, which does not rely on the assumption of the states being pure.

\begin{table*}[t]
\centering
\renewcommand{\arraystretch}{1.3} % Row height
\setlength{\tabcolsep}{3pt}       % Column separation

\begin{tabular}{lllccccc}
\hline\hline
\multicolumn{2}{c}{Ions} & State & \multicolumn{3}{c}{Fidelity to ideal state} & Purity & Relative overlap \\
Alice & Bob & & PM & PT & 2$n$+1-basis method & &  \\
\hline
${}^{88}\mathrm{Sr}^+$ & ${}^{88}\mathrm{Sr}^+$ &
$\frac{1}{\sqrt{2}}\ket{01}+\ket{10}$ &
96.0(7) \% & 96.94(9) \% & 98.72(3) \%  & 94.1 \% & 95.8 \%  \\

${}^{88}\mathrm{Sr}^+$ & ${}^{43}\mathrm{Ca}^+$ &
$\frac{1}{\sqrt{2}}\ket{00}+\ket{11}$ &
95.1(8) \% & 94.1(6) \% & 96.61(5) \%  & 89.0 \% & 89.6 \%  \\

${}^{43}\mathrm{Ca}^+$ & ${}^{43}\mathrm{Ca}^+$ &
$\frac{1}{\sqrt{2}}\ket{00}+\ket{11}$ &
92(1) \% & 93.1(7) \% & 92.85(4) \%   & 87.0 \% & 91.5 \% \\

${}^{88}\mathrm{Sr}^+$ & ${}^{88}\mathrm{Sr}^+$, ${}^{43}\mathrm{Ca}^+$ &
$\frac{1}{\sqrt{2}}\ket{000}+\ket{111}$ &
94(1) \% & 93.1(7) \% & 94.96(7) \%  & 87.3 \% & 86.6 \%  \\

\hline\hline
\end{tabular}

\caption{\justifying \textbf{Summary of mixed-species trapped ion remotely entangled states characterised in this work.}
This table includes the observed entanglement fidelities with respect to the ideal states extracted using parity measurements (PM), Pauli tomography (PT) and the 2$n$+1-basis method introduce in this article.
The density matrices extracted using PT and the new 2$n$+1-basis method can be found in Appendix \ref{SM:IonTrap}. The purity of the experimental states is calculated from the density matrices estimated from PT. The relative overlap is calculated between the density matrices estimated using PT and the 2$n$+1 method.}
\label{tab:summary}
\end{table*}

\vspace{1cm}

\section{Conclusions}

We introduced two methods to estimate the global quantum state of an $n$ qubit quantum computer. The first one considers five measurements bases only, independently of the number of qubits. Three of such bases are tensor product, and the resting two bases do not have entangled states but require classical communication between qubits to be implemented. The second method is composed of $2n+1$ tensor product bases, thus not requiring classical communication. Each measurement stage is realized by a quantum circuit that applies a single local gate to one qubit, followed by measurement in the computational basis. \textit{This configuration reaches the absolute limit of simplicity for a measurement stage, while being fully scalable and compatible with current quantum hardware.} Hence, both estimation methods do not require the use of entangling quantum gates. The absence of such gates in the measurement stage minimizes error accumulation, allowing for higher fidelity estimation across a larger number of qubits. Scalability, simplicity, and higher fidelity are three notable advantages of our methods over existing estimation methods. An additional advantage is that the estimates are obtained analytically from the measurement results; that is, there is no need to solve a system of equations or to use statistical inference for data post-processing. Thus, the classical computational cost of our methods is reduced to a minimum, thereby preserving scalability. Furthermore, an important feature is the possibility of certifying that the prepared state is pure or nearly pure; see Appendix \ref{SM:purity}.    

The performance of both methods was tested on freely accessible IBM quantum computers using genuinely entangled quantum states. The $2n+1$-basis method achieved a high fidelity of $f \approx 0.93$ for 12 qubits, exhibiting unprecedented results on quantum computers, whereas the five-basis disentangled method yielded $f \approx 0.94$ for seven qubits. We also compared the two methods on the same quantum processor. First, we compared the results of both methods with a fixed number of experiments for each basis, demonstrating that in this case the $2n+1$ bases method performs better. Subsequently, the measurements were repeated, keeping the total size of the ensemble fixed; in this context, the 5 bases method demonstrated superior performance. The advantage of the 5 bases method for a fixed ensemble size comes from the fact that having fewer outcomes then fewer repetitions are required to accurately estimate the probabilities distribution for the 5 bases method. On the other hand, if we are not limited by the size of the ensemble, it is better to implement the $2n+1$ bases method, whose advantage comes from the fact that by having a larger number of outcomes, the method is more robust to errors. We have also demonstrated our methods in two remote 2-qubit ion traps by reconstructing a distributed entangled state, showing the possibility of application in distributed quantum computing. 

The proposed protocols offer a scalable, efficient and reliable method for benchmarking quantum states, suitable for current quantum hardware. Depending on the system characteristics, one approach may outperform the other or even work best in combination. For example, an initial characterization can be performed with $2n+1$ product measurements, and once the optimal estimation path is identified, the $5$ bases protocol can be applied to enhance statistics while reducing sample costs. %However, we believe that there may still be room for optimization and, possibly, incorporation of $2n+1$ or $4+1$ protocols as part of adaptive estimation methods.

%Finally, the results presented in this work are restricted to multi-qubit systems, which are the primary testbed for future quantum technologies; the limitation is also due to the fact that polarization identity naturally implies 2-outcome measurements which are natural for 2-level system. %In an upcoming theoretical paper~\cite{czartowski2025generalpolarization} we show that it is possible to present a similar notion for a general qudit system, which is a necessary ingredient for local estimation of pure states of distributed systems with higher local dimensionality.

% \jcz{LEFT OUT! Dardo, do you think you can draft a first version?}
% \jcz{\textbf{Correction} -- \textbf{Victor}, I am leaving this to you}

% \section{Acknowledgments}
\acknowledgments

% \jcz{Guys, any additional acknowledgments?}
DG acknowledges financial support from FONDECYT Regular Grant No. 1230586 (Chile). VGA and KDO are enrolled in the PhD Program Doctorado en Física, mención Física-Matemática at the Universidad de Antofagasta, Chile. VGA further acknowledges support from the ANID doctoral fellowship No. 21221008. JCz is supported by the start-up grant of the Nanyang Assistant Professorship at Nanyang Technological University, Singapore, awarded to Nelly Ng. GA acknowledges support from Wolfson College, Oxford, and Cisco. DM is currently employed by Oxford Ionics. AD is supported by the ANID Millennium Science Initiative Program ICN17$_-$012. We acknowledge the use of IBM Quantum services for this work. The views expressed are those of the authors, and do not reflect the official policy or position of IBM or the IBM Quantum team.

\appendix

\section{Proof of Proposition \ref{prop:cnot_prop}}
\label{app:CNOT_separability}

In this appendix, we present a complete proof of the following proposition:

\cnotprop*

We begin by demonstrating the following lemma:

\begin{lem} \label{lem:ham_1_dist}
    Consider a pair of binary strings $\vb{j}',\,\vb{k}'$, defined by $\ket{\vb{j}'} = S\ket{\vb{j}},\,\ket{\vb{k}'} = S\ket{\vb{j}}$ with $S$ given by \eqref{cnots}. The Hamming distance between $\vb{j}'$ and $\vb{k}'$ is given by
    \begin{equation} \label{eq:lem_ham_1_statement}
        d_H(\vb{j}',\vb{k}') = 1.
    \end{equation}
\end{lem}

\begin{proof}
% As we will be dealing only with properties of the computational states $\ket{\vb{j}}$ and a chain of controlled NOTs $S = \CNOT_{n-1,n-2}\hdots\CNOT_{1,0}$, \bl{mapping computational states to computational states}, in what follows we will consider equivalently binary numbers $\vb{j}$ less than $2^n - 1$ 
% % and their respective binary representation $j_{n-1}\hdots j_0$ 
% and limit ourselves to basic binary arithmetic \dg{Too complicated phrase. What is the message here? That $S$ is a permutation so the computational basis remains unentangled?}.
% % \dg{I think there is no need to clarify the standard notation followed along the entire paper}. 
% \bl{For convenience, let us also set $\vb{k} = \vb{j} + 1$ and primed numbers $\vb{j}',\,\vb{k}'$ by relations
% \begin{align}
%     \ket{\vb{j}'} = S\ket{\vb{j}}, && 
%     \ket{\vb{k}'} = S\ket{\vb{k}}.\label{eq:CNOTted_strings}
% \end{align}
% \oli{By Lemma \ref{lem:ham_1_sep},} we know that $\ket{\vb{j}'} + \ket{\vb{k}'}$ is separable if and only if $d_H(\vb{j}',\vb{k}') = 1$.%, which is to be demonstrated for all~$\vb{j}$.
% } 

Recall the natural mapping $\vb{j}\mapsto \sum_i j_i 2^i$, which will allow us to interpret $\vb{j}$ as both a binary string and the corresponding integer in the binary representation. We begin by explicitly stating the relation between the digits of $\vb{j}$ and its successor, $\vb{k} = \vb{j}+1$, using the standard carry-over concept,
\begin{equation}\label{addition}
    \begin{aligned}
        k_0 & = j_0 \oplus 1, & c_1 & = j_0, \\
        k_1 & = j_1 \oplus c_1, & c_2 & = j_1\odot c_1, \\
        \hdots \\
        k_{n-2} & = j_{n-2} \oplus c_{n-2}, & c_{n-1} & = j_{n-2}\odot c_{n-2}, \\
        k_{n-1} & = j_{n-1}\oplus c_{n-1},
    \end{aligned}
\end{equation}
where we used $\oplus$ and $\odot$ for addition and multiplication in $\mathbb{Z}_2$ for full clarity, %\footnote{Note that multiplication $\odot$ does not differ from ordinary multiplication, and introduce the symbol only for notational consistency.}
 and overflow property $\vb{j} = 2^n-1\Rightarrow\vb{k} = 0$ is realized due to fixed number $n$ of digits.  We immediately resolve the recurrence for the carry-over by setting

\begin{equation} \label{eq:carryover_closedform}
    c_\ell = j_{\ell - 1}\odot c_{\ell - 1} =  \bigodot_{i=0}^{\ell-1} j_i
\end{equation}
which imposes, for consistency, $c_0 = 1$ and can be extended all the way to $\ell = n$. 
% \dg{$\ell = n-1$?}\jcz{That's the thing, no! It is an important detail, as $c_n$ is useful later in the compact form of \eqref{eq:shift_ham_compact}}. 
In turn, the Hamming distance is given explicitly as
\begin{equation}
    d_H(\vb{j}, \vb{k}) = \sum_{i=0}^{n-1} \abs{j_i - k_i} = \sum_{i=0}^{n-1} c_i
\end{equation}

% \dg{How do you arrive to (J4)? Can you provide a small hint to the proof?}\jcz{But this is just a direct calculation, replacing $j$ and $k$ from previous definitions} 
From (\ref{eq:carryover_closedform}) we observe $c_\ell = 1$ if and only if all $j_i = 1$  for $i \leq \ell$. As a consequence, after simple algebraic manipulation we find that
\begin{equation}
    d_H(\vb{j}, \vb{k}) = \max \ell+1: \vb{k} \equiv 0\mod 2^\ell.
\end{equation}
% \dg{Is it possible to provide at least a hint about how to derive (J3)?}
This statement can be made explicit -- adding 1 to $\vb{j}$ ending in a chain of ones of length $\ell_*$ results in a carry-over cascade until $j_{\ell_*} = 0$, resulting in $\vb{k}$ having $\ell_*$ final digits equal to 0 and $k_{\ell_*} = 1$, making it divisible by $2^{\ell}$ for all $\ell \leq \ell_*$; it is worth noting, that the reasoning holds also
% \dg{`even' replaced by `also'} 
for odd $\vb{j}$, ie. if $\ell_* = 0$.

Now, let us limit ourselves to first two digits, setting $n = 2$ 
% and apply a classical \dg{classical?} \jcz{Or logical, whichever suits better -- but this is a relatively standard notion for bitstrings, no controversy about its action. Although, I think \textbf{the word can be skipped altogether}} $\CNOT_{0,1}$ to both $\vb{j}$ and $\vb{k}$, defining $\vb{j}' = \CNOT_{0, 1}\vb{j}$ and $\vb{k}'$ in analogy. This yields digits
and consider digits of $\vb{j}'$ and $\vb{k}'$ as defined in the statement of the lemma, That is,
\begin{equation}
    \begin{aligned}
        j'_0 & = j_0 \oplus j_1, & k'_0 & = j_0 \oplus 1 \oplus (j_1 \oplus j_0) \\
        &&&= j_1\oplus 1, \\
        j'_1 & = j_1, & k'_1 & = j_1 \oplus j_0 
    \end{aligned}
\end{equation}
We can evaluate Hamming distance as
\begin{equation}
    d_H(\vb{j}',\vb{k}') = (j_0 \oplus 1) + j_0 = 1
\end{equation}
which is in line with the proposition.

Extending this to $S$
% \dg{I think it is more precise to mention $S$ instead of sequence of descending CNOTs}
for arbitrary $n \geq 2$, we find
\begin{equation}
    \begin{aligned}
        j'_\ell & = j_\ell \oplus j_{\ell+1}, & k'_\ell & = j_\ell \oplus c_\ell \oplus j_{\ell+1} \oplus c_{\ell+1} \\
        &&& = j_\ell \oplus j_{\ell+1} \oplus c_\ell\odot(j_\ell\oplus 1), \\
        j'_{n-1} & = j_{n-1}, & k_{n-1} & = j_{n-1} \oplus c_{n-1}
    \end{aligned}
\end{equation}
with $\ell < n-1$.

From this, we see that Hamming distance is given by
\begin{equation} \label{eq:shift_ham_dist}
    d_H(\vb{j}',\vb{k}') = c_{n-1} + \sum_{\ell = 0}^{n-2} c_\ell\odot(j_\ell \oplus 1).
\end{equation}
% From the above the reasoning, we may put forward the following reasoning:
% \begin{enumerate}
%     \item Suppose $c_{n-1} = 1$, which implies $j_\ell = 1$ for $\ell < n-1$, which sets $j_\ell \oplus 1 = 0$, thus $d_H(\vb{j}',\vb{k}') = 1$.
%     \item Else, $c_{n-1} = 0$ and $c_{n-2} = 1$, which implies $j_\ell = 1$ for $\ell < n-2$, which sets $j_\ell \oplus 1 = 0$, thus $d_H(\vb{j}',\vb{k}') = 1$.
%     \item[$\hdots$]
%     \item[$m$.] Else, $c_{n-m+1} = 0$ and $c_{n-m} = 1$, which implies $j_\ell = 1$ for $\ell < n-m$, which sets $j_\ell \oplus 1 = 0$, thus $d_H(\vb{j}',\vb{k}') = 1$.
%     \item[$\hdots$]
%     \item[$n-1$.]  Else, $c_{2} = 0$ and $c_{1} = 1$, which implies $j_\ell = 1$ for $\ell < n-m$, which sets $j_\ell \oplus 1 = 0$, thus $d_H(\vb{j}',\vb{k}') = 1$.
%     \item[$n$.] Else, $c_{1} = 0$. In this case $\vb{j}' = 0,\,\vb{k}' = 1$ and thus $d_H(\vb{j}',\vb{k}') = 1$.
% \end{enumerate}

% The above steps exhaust all possibilities, and thus provide a proof that $d_H(\vb{j}',\vb{k}') = 1$ in general.

We may set $j_n = 0$ formally, as being beyond the range of an $n$-digit binary number. With this, we arrive at the following %a more 
compact expression of \eqref{eq:shift_ham_dist}:%for 
\begin{equation} \label{eq:shift_ham_compact}
    d_H(\vb{j}',\vb{k}') = \sum_{m = 0}^{n-1} c_m\odot(j_m \oplus 1 = \sum_{m=0}^{n-1}c_{m+1} \oplus c_m
\end{equation}
which is clear from \eqref{eq:shift_ham_dist} and \eqref{eq:carryover_closedform}. Let us take $\ell_*$ such that $j_{\ell_*} = 0$ and $j_m = 1$ for all $m < \ell_*$. This implies 

    \begin{equation}
    c_m = \Theta(\ell_* - m) \Rightarrow c_m \oplus c_{m+1} = \delta_{m,\ell_*}\label{eq:simple_proof}
    \end{equation}
where $\Theta$ is the Heaviside step function with the assumption that $\Theta(0) = 1$. This, together with \eqref{eq:shift_ham_compact}, leads to $d_H(\vb{j}',\vb{k}') = 1$, which concludes the proof.
\end{proof}

From this Lemma follows a simple property
\begin{obs}\label{obs:sep_under_S}
    Consider $a, b\in\mathbb{C}$. A (unnormalized) state $S\qty(a\ket{\vb{j}} + b \ket{\vb{j}+1})\in\mathcal{H}_2^{\otimes n}$, with $S$ defined by \eqref{cnots}, is a product state, for any $\vb{j}=0,\dots,2^{n-1}$.
\end{obs}
\begin{proof}
    This follows directly by applying Lemmas \ref{lem:ham_1_sep} and \ref{lem:ham_1_dist}
\end{proof}

Applying Observation \ref{obs:sep_under_S} sequentially to all the states of  $\tilde{\mathfrak{B}}_i$ for $i = 1, 2, 3, 4$ and observing invariance of $\mathfrak{B}_0$ under $S$ completes the proof of Proposition \ref{prop:cnot_prop}. \qed

\section{Alternative construction of disentangled bases}
\label{sec:LOCC_prot}

The five measurement bases introduced in Section \ref{sec:5bases} can be implemented using local measurements and classical communication. To this end, below we present an alternative derivation of the 5 disentangled bases, equivalent to the one described beforehand up to irrelevant relabeling of the local qubit states and the qubit numbering. The first basis $\tilde{\mathfrak{B}}_0$ is the computational basis and therefore can be omitted.  

The next two measurement bases $\tilde{\mathfrak{B}}_1, \tilde{\mathfrak{B}}_2$ correspond to $H\otimes\mathbb{I}_2^{\otimes n-1}$ and $\widetilde{H}\otimes\mathbb{I}_2^{\otimes n-1}$, respectively, where we define $H = \op{+}{0} + \op{-}{1}$ and $\widetilde{H} = \op{+i}{0} + \op{-i}{1}$. It is beneficial to visualize it in terms of graph $\mathcal{G}$ -- this set of measurements allows us to resolve polarization identities for a complete set of parallel edges connecting vertices with $i_1 = 0$ to the respective vertices with $i_1 = 1$. In terms of heuristic, this provides us with a scaffolding, reducing the number of disconnected parts of the graph by half to $2^{n-1}$.

\begin{figure}[h!]
    \centering
    \includegraphics[width=0.8\linewidth]{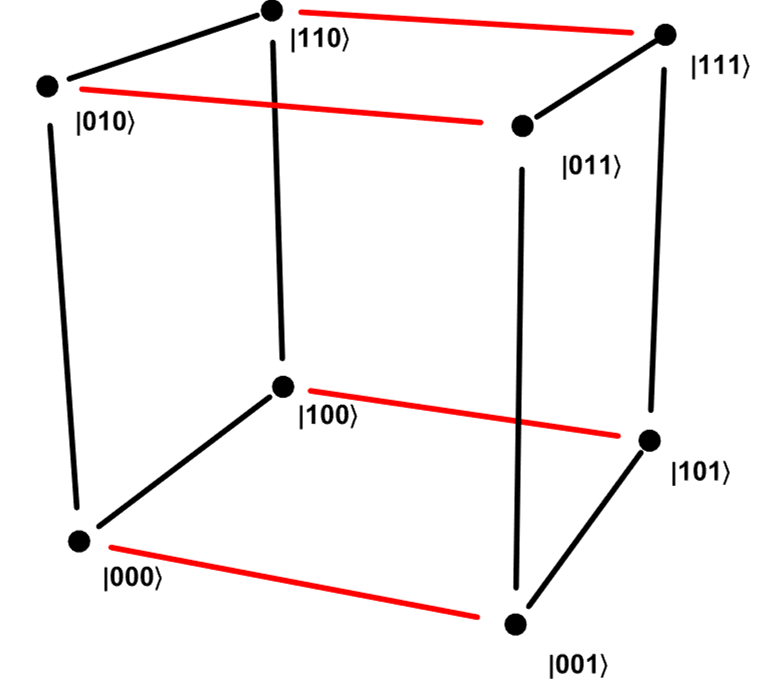}
    \caption{\justifying (Color online) Degree $3$ hypercubic graph with scaffolding edges highlighted in red.}
    \label{fig:hypercube_3q_se}   %%%%
\end{figure}

The construction of the remaining two measurement bases $\tilde{\mathfrak{B}}_3, \tilde{\mathfrak{B}}_4$ will follow in steps, resulting in a conditional structure of the basis. First, we want to further reduce the number of connected components of $\mathcal{G}$. We could easily half them by considering the measurement corresponding to the bases $\mathbb{I}_2\otimes G\otimes\mathbb{I}_2^{\otimes n-2}$ with $G = H, \widetilde{H}$; this is, however, excessive, as each pair of previously disconnected components is joined by a pair of edges. Thus, we fix the operator under consideration to be $O_{G,1} = \op{1}\otimes G\otimes\mathbb{I}_2^{\otimes n-2}$ and take its nonzero columns to define the first $2^{n-1}$ projectors for each of our measurements. In terms of logical structure, this corresponds to saying measure $\sigma_{X/Y}$ on the second qubit and $\sigma_Z$ on the remaining $n-2$ qubits if the result of measurement $\sigma_Z$ on the first qubit was 1 or, in terms of graph $\mathcal{G}$, we connected the vertices with $i_1 = 1, i_2 = 1$ with the corresponding vertices for $i_1, i_2 = 1$. Subsequent steps follow in a similar manner, resulting in a set of operators
\begin{equation}
    \begin{alignedat}{3}
        &O_{G,1} = & &&&\op{1}\otimes G\otimes\mathbb{I}^{\otimes n-2} \\
        &O_{G,2} = & &\op{0}&\otimes&\op{1}\otimes G\otimes\mathbb{I}_2^{\otimes n-3} \\
        &O_{G,n-2} = & &\op{0}^{\otimes n-3} &\otimes&\op{1}\otimes G \otimes \mathbb{I}_2 \\
        &O_{G,n-1} = & &\op{0}^{\otimes n-2} &\otimes&\op{1}\otimes G
    \end{alignedat}
\end{equation}
with $G = H,\widetilde{H}$. All of the above operators add up to a pair of bases
\begin{equation}
    \begin{aligned}
        \tilde{\mathfrak{B}}_3 & = \sum_{i=1}^{n-1} O_{H,i}, &
        \tilde{\mathfrak{B}}_4 & = \sum_{i=1}^{n-1} O_{\widetilde{H},i}
    \end{aligned}
\end{equation} 
The preparation procedure for the five bases $\tilde{\mathfrak{B}}_i$ can also be algorithmically enclosed in a simple conditional procedure, presented in Algorithm \ref{alg:cond_proc} and the logical circuit is presented in Fig.~\ref{fig:LOCC_Circuit}.

\begin{algorithm}[H]
\caption{Implementation of the (4+1) basis tomographic scheme} \label{alg:cond_proc}
\KwIn{$\ket{\psi}\in\mathcal{H}_2^{\otimes n}$ \tcp*{Pure input state}}
\KwIn{$b_C, b_S, b_M, b_H \in \{\text{True},\text{False}\}$ \tcp*{Control bits}}
\KwIn{$N\in\mathbb{N}$ \tcp*{Number of rounds}}
\KwIn{$\vb{c}\in\{0,1\}^{n\times N}$ \tcp*{Memory}}
\KwIn{$i,j \in \mathbb{N}$ \tcp*{Loop variables}}
\KwOut{Updated memory $\vb{c}$ after measurements}

$i \gets 1$\;
\While{$i \le N$}{
    $b_C \gets \text{rand}(\text{True},\text{False})$\;
    $b_S \gets \text{rand}(\text{True},\text{False})$\;
    $b_M, b_H \gets \text{False}$\;
    $j \gets 1$\;
    \While{$j \le n$}{
        \If{($j = 1 \wedge b_C) \vee (j \le (n-1) \wedge \neg b_H \wedge b_M) \vee (j=n \wedge \neg b_H)$}{
            $\ket{\psi} \gets H\ket{\psi}$\;
            $b_H \gets \text{True}$\;
            \If{$b_S$}{
                $\ket{\psi} \gets S\ket{\psi}$\;
            }
        }
        $b_M \gets \text{True}$\;
        $c_{ij} \gets \operatorname{measure}_j \ket{\psi}$ \tcp*{Measure $j$th subsystem}
        $j \gets j + 1$\;
    }
    $i \gets i + 1$\;
}

\Return{$\vb{c}$}\;

\end{algorithm}

\section{Optimization Algorithm}
\label{app:selection_procedure}

% \jcz{I will be reading and correcting this fragment!}

When considering real-world implementation of a quantum circuit, a well-known natural tension arises between abstract architecture-agnostic formulation of an algorithm and architecture constraints and gate errors stemming from the actual physical hardware.

Hence, the main goal of optimizing implementation is to ensure that non-local gates are placed in accordance with the physical constraints of the available hardware (for example, in some cases gates can be implemented from qubit one to qubit two, but not from two to one) and to use the qubits with the lowest possible gate error rates. The mathematical modeling of this problem can be represented as a graph. In this graph, the weights associated with the edges represent error rates associated to the nonlocal gates, and the node weights correspond to the error associated with local gates, both information that can be extracted from the hardware of the quantum processor.

The central challenge is how to efficiently select a subgraph spanned by $n$ nodes providing an optimal path. This problem is not new; in fact, it is a specific instance of the well-known \textit{Traveling Salesman Problem} (TSM). The algorithm designed to address this problem proceeds as follows:

\subsection*{Optimization Algorithm}

\begin{enumerate}
    \item Fix an initial node $v_0$.
    \item Identify all nodes reachable within $n$ steps from $v_0$, using the Dijkstra's algorithm~\cite{dijkstra}.
    \item Compute the total weight (or cost) of each of the resulting paths.
    \item Repeat the process for all nodes in the graph.
    \item Compare the total weights and select the path that minimizes the overall cost.
\end{enumerate}

It is worth noting that if the graph is directed, the computational complexity of the algorithm is significantly reduced, as the edge orientation introduces additional constraints that improve the efficiency of pathfinding.

%\begin{algorithm}[t]
%\caption{Find the optimal path with minimum weight at fixed distance~\cite{dijkstra}}
%\begin{algorithmic}[1]

%\Require Graph $G$
%\Require Distance $D$ \Comment{Quantity of Qubits}
%\Ensure Path of length $D$ with minimal total weight

%\State minWeight $\gets \infty$
%\State bestPath $\gets$ None
%\State Paths $\gets \emptyset$

%\For{each node in $G$}
    %\State result $\gets$ SingleSourceDijkstra($G$, node)
    %\For{each path in result.paths}
%        \If{length(path) = $D$}
%            \State Paths $\gets$ Paths $\cup \{ path \}$
%        \EndIf
%    \EndFor
%\EndFor

%\For{each path $P$ in Paths}
%    \State $w \gets$ PathWeight($G$, $P$)
%    \If{$w < minWeight$}
%        \State minWeight $\gets w$
%        \State bestPath $\gets P$
%    \EndIf
%\EndFor

%\State \Return bestPath

%\end{algorithmic}
%\end{algorithm}

\begin{algorithm}[H] % [H] = Aquí fijo, sin flotar
\caption{Find the optimal path with minimum weight at fixed distance}
\KwIn{Graph $G$, distance $D$ \tcp*[r]{Quantity of Qubits}}
\KwOut{Path of length $D$ with minimal total weight}

$minWeight \gets \infty$\;
$bestPath \gets$ None\;
$Paths \gets \emptyset$\;

\ForEach{node in $G$}{
    \hspace{0.4cm}result $\gets$ SingleSourceDijkstra($G$, node)\;
    \hspace{0.4cm}\ForEach{path in result.paths}{
        \hspace{0.8cm}\If{$length(path) = D$}{
            \hspace{1.2cm}$Paths \gets Paths \cup \{ path \}$\;
        }
    }
}

\ForEach{path $P$ in $Paths$}{
    \hspace{0.4cm} $w \gets$ PathWeight($G$, $P$)\;
    \hspace{0.4cm} \If{$w < minWeight$}{
        \hspace{0.8cm}$minWeight \gets w$\;
        \hspace{0.8cm}$bestPath \gets P$\;
    }
}

\Return{$bestPath$}\;

\end{algorithm}

\subsubsection*{Results and Validation}

% \jcz{I think it might be good to conduct randomized testing as well -- it does not have to be performed on hardware directly, but just on the noise model -- these are readily available, and one can simulate small but randomized circuits. This would help avoid potential objections that the preparation we are using may be a special case, not extending to generic ciruits...}

The algorithm was tested on the same states described in section V. The results are shown in Figure~\ref{fig:hellinger}. In the plot, the $y$-axis represents the \textit{Hellinger distances} between the experimentally obtained probability distributions and the ideal theoretical distribution. Control refers to the experiment sent to a randomized sequence of $n$ qubits with no further optimization. 
% \jcz{Now I see that this choice of control is not optimal, as it might not be representative of the performance of the entire machine. What I would suggest is to randomly select a subset of $n$ qubits in such a way that they form a connected subgraph of architecture via nearest-neighbour edges. I think this kind of averaged control can be evaluated using local model of the QPU, as it would be very shot-costly... But this randomization would actually reflect overall performance of the chip as oposed to the first $n$ qubits, out of which some may be faulty, already introducing an unwanted bias} \tmv{I agree with your suggestion. I ran the test, and the results were actually better this time. With TSM optimization, we achieved better performance than level 3 optimization in all cases. However, it's still the same quantum state — I'm currently working on randomizing the input states}
% \jcz{Great, looking forward to the results!}
%  As can be seen, the proposed optimization method outperforms IBM’s native optimization in nearly all tested cases. 
 % \jcz{An open question -- should we communicate with IBM to implement this code as part of their routines? Seems that it may be a good thing...}

\begin{figure}[h]
    \centering
    \includegraphics[width=0.5\textwidth]{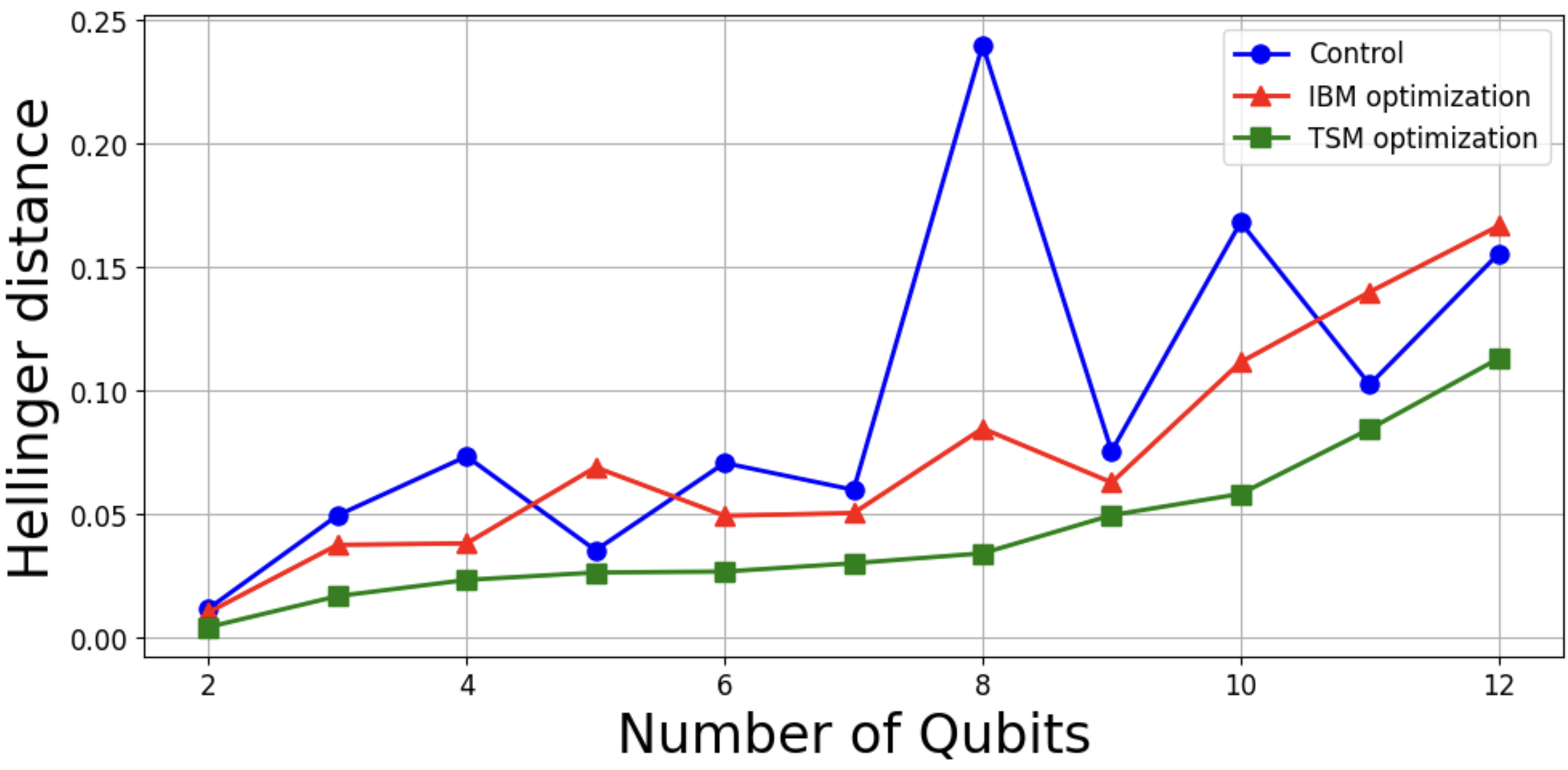}
    \caption{\justifying Comparison of Hellinger distances for each optimization strategy. Lower values indicate higher fidelity with respect to the theoretical distribution}
    \label{fig:hellinger}
\end{figure}

\section{Seven informationally complete bases for 3 qubits}\label{SM:7bases}

We provide the explicit expression of the 5 and $2n+1$ informationally complete bases for $n=3$ qubits. Let us start with the 5 bases:
$$\mathfrak{B}_0=\{|000\rangle,|001\rangle,|010\rangle,|011\rangle,|100\rangle,|101\rangle,|110\rangle,|111\rangle\},$$ $$\mathfrak{B}_1=\{|00,\pm\rangle,|01,\pm\rangle\,|10,\pm\rangle,|11,\pm\rangle\},$$ $$\mathfrak{B}_2=\{|00,\pm i\rangle,|01,\pm i\rangle\, |10,\pm i\rangle,|11,\pm i\rangle\},$$ 
$$\mathfrak{B}_3=\{|000\rangle\pm|100\rangle, |001\rangle\pm|011\rangle,$$
$$|010\rangle\pm|110\rangle, |111\rangle\pm|101\rangle\}$$
$$\mathfrak{B}_4=\{|000\rangle\pm i|100\rangle, |001\rangle\pm i|011\rangle,$$
$$|010\rangle\pm i|110\rangle, |101\rangle\pm i|111\rangle\}$$

Three qubit systems is the smallest case for which the $2n+1$ bases differs from the separable 5 bases. We recall that the $2n+1$ bases arise as a solution (\ref{states_pol}) under the full separability constraint established in Lemma \ref{lem:ham_1_sep}, where 48 tensor product vectors arise for $n=3$, that can be sorted into 6 orthonormal bases in a unique way. Complemented by the computational basis $\mathfrak{B}_0$, these bases are defined as follows:

$$\mathfrak{B}_0=\{|000\rangle,|001\rangle,|010\rangle,|011\rangle,|100\rangle,|101\rangle,|110\rangle,|111\rangle\},$$ $$\mathfrak{B}_1=\{|00,\pm\rangle,|01,\pm\rangle\,|10,\pm\rangle,|11,\pm\rangle\},$$ $$\mathfrak{B}_2=\{|00,\pm i\rangle,|01,\pm i\rangle\, |10,\pm i\rangle,|11,\pm i\rangle\},$$ 
$$\mathfrak{B}_3=\{|0,\pm,0\rangle,|0,\pm,1\rangle, |1\pm,0\rangle,|1\pm,1\rangle\},$$ 
$$\mathfrak{B}_4=\{|0,\pm i,0\rangle,|0,\pm i,1\rangle, |1,\pm i,0\rangle,|1,\pm i,1\rangle\},$$
$$\mathfrak{B}_5=\{|\pm,00\rangle,|\pm,01\rangle, |\pm,10\rangle,|\pm,11\rangle\},$$ 
$$\mathfrak{B}_6=\{|\pm i,00\rangle,|\pm i,01\rangle, |\pm i,10\rangle,|\pm i,11\rangle\}.$$\vspace{0.5cm}

% \section{Temporal appendix: CNOTs chain}

% Sea $$U = \CNOT_{n-1,n} \cdots \CNOT_{2,3} \cdot \CNOT_{1,2},$$
% donde $\CNOT_{i,j}$ es una compuerta controlada con control en el qubit $i$ y objetivo en $j$. La acción de $\CNOT_{i,j}$ sobre un estado base es:
% $$ \CNOT_{i,j} \left|x_1 \cdots x_n\right\rangle = \left|x_1 \cdots x_j \oplus x_i \cdots x_n\right\rangle $$
% Acción del operador $U$ completo
% $$ U \left| x_1 x_2 \cdots x_n \right\rangle = \left| y_1, y_2, \ldots, y_n \right\rangle $$
% donde $y_k = x_1 \oplus x_2 \oplus \cdots \oplus x_k$, para $k = 1,\ldots,n$.
% Expresión analítica del operador
% $$ U = \sum_{x_1,\dots,x_n \in \{0,1\}} \left| y_1 y_2 \cdots y_n \right\rangle \left\langle x_1 x_2 \cdots x_n \right| $$
% con $y_k = x_1 \oplus x_2 \oplus \cdots \oplus x_k$.\vspace{0.5cm}

% Interpretación: Este operador corresponde a una matriz de permutación que reordena los estados base de forma que cada bit de salida es la XOR acumulada de los bits previos. 

\section{Graph optimization}
\label{app:errors}
        
Before moving on to discuss the process of error minimization, let us introduce an additional graph-theoretic definition
    \begin{defn}
        Consider a graph $\mathcal{G} = \qty{V, E}$ and a pair of vertices $\vb{j},\,\vb{k}\in V$. A connected subgraph $P^{(\vb{jk})} = \qty{V^{(\vb{jk})},E^{(\vb{jk})}}\subset \mathcal{G}$ is called \textbf{a path} between $\vb{j}$ and $\vb{k}$ if $\vb{j},\,\vb{k}\in V^{(\vb{jk})}$, and the vertices $\vb{l} \notin \qty{\vb{j},\vb{k}}$ have exactly two neighbors, and $\vb{j},\,\vb{k}$ have only one neighbor each.
    \end{defn}
    % \dg{\textbf{Jakub}. Please extend this definition, and the coming paragraphs, to include cycles. Abraham bases define cycles.} \jcz{Still, for the purpose of estimation full cycle is not used, but rather paths contained in a cycle -- hence, the cycle is never invoked for the purpose of estimation.}
    
    \begin{defn}
        A connected subgraph $C = \qty{V,E}\subset \mathcal{G}$ is called \textbf{a cycle} if each vertex $\vb{v}\in V$ has exactly two neighbors.
        % \jcz{Interestingly, this minimal definition encodes simple intuition -- any node needs to have a way in and a way out}
    \end{defn}
    
    \noindent 
    
    Note that this definition matches an intuition that endpoints of a path should have only outgoing edges -- hence a single neighbor -- while all the midpoints of the path should have incoming and outgoing edges, hence exactly two neighbors. For convenience, we refer to a set of all paths between $\vb{j}$ and $\vb{k}$ as $\operatorname{paths}^{(\vb{j}\vb{k})}$. 

    The number of edges in a path, $\abs{E^{(\vb{jk})}} = \abs{V^{(\vb{jk})}} - 1$, has a natural interpretation as the length of a path, according to a discrete metric. With this, we can consider a set $\mathcal{P}$ of paths of minimal length, defined formally as
    \begin{equation}\label{eq:min_lenght_paths}
        \mathcal{P}^{(\vb{jk})} = \qty{P^{(\vb{jk})}\in\operatorname{paths}^{(\vb{j}\vb{k})}: \abs{E^{(\vb{jk})}} = \!\!\!\!\min_{\operatorname{paths}^{(\vb{j}\vb{k})}} \abs{E^{(\vb{jk})}}} .
    \end{equation}

    % \vdga{
    % Before discussing the process of error minimization, let us introduce the following    
    % \begin{defn}
    %     Consider a state $\ket{\psi_n}\mathcal{H}_2^{\otimes n}$ and its corresponding graph $\mathcal{G}\qty(V, E)$. Then, a path $P^{(\vb{jk})}\qty(V^{(\vb{jk})},E^{(\vb{jk})})\subset \mathcal{G}$ is a subgraph connecting the nodes $\vb{j}$ and $\vb{k}$, where $V^{(\vb{jk})}\subset V$ represent the nodes of $P$ and $E^{(\vb{jk})}\subset E$ contains all the necessary edges to connect $\vb{j}$ and $\vb{k}$.
    % \end{defn}
    % Also, we can define the set 
    % $$ \mathcal{P}=\qty{P^{(\vb{jk})}\qty(V^{(\vb{jk})},E^{(\vb{jk})}),\min_{\abs{E^{(\vb{jk})}}}P^{(\vb{jk})}}
    % $$
    % as the set containing all paths connecting nodes $\vb{j}$ and $\vb{k}$ through the lower possible number of edges. 
    % }
    
The hypercube is a Hamiltonian graph, so it is always possible to find a path between any pair of its nodes. Furthermore, the variance of the amplitudes $|\langle\psi|\vb{j}\rangle|$ and $|\langle\psi|\vb{j}'\rangle|$ involved in the polarization identity (\ref{pol_identity}) are proportional to the weights $w_{\vb{j}}=|\langle\psi|\vb{j}\rangle|^2$ and $w_{\vb{j}'}=|\langle\psi|\vb{j}'\rangle|^2$, respectively. In addition, the variance in the complex phase of the coefficient $a_{\vb{j}}=\langle\psi|\vb{j}\rangle$ depends on the distance between the vertices $v_{\vb{j}}$ and $v_{\vb{j}'}$ along the graph. See Appendix \ref{apx:error_propagation} for details. In order to minimize error propagation, we aim to find a path with minimum length between $v_{\vb{j}}$ and $v_{\vb{j}'}$ but also considering the largest possible weights along the path. That is,
\begin{equation}\label{eq:optimal_rezo}
\mathcal{P}_{opt}^{(\vb{j}\vb{j}')}= \max_{\mathcal{P}^{(\vb{j}\vb{j}')}}\qty{\sum_{\vb{l}} \abs{a_{\vb{l}}}^2; \vb{l}\in V^{(\vb{j}\vb{j}')}},
\end{equation}
where $\{\mathcal{P}^{(\vb{j}\vb{j}')}(V^{(\vb{j}\vb{j}')},E^{(\vb{j}\vb{j}')})\}$ is the set of paths from vertex $v_{\vb{j}}$ to vertex $v_{\vb{j}'}$ that has the minimum possible length, as defined in \eqref{eq:min_lenght_paths}. For a detailed description on error propagation, see Appendix \ref{apx:error_propagation}

\section{Analytic tomographic reconstruction}\label{app:analityc_sol}

Let \(\mathcal{G}=(V,E)\) be a chosen connected graph with \(|V|=2^{n}\) vertices. To reconstruct the state, the following coupled non-linear system of equations should be solved:
\begin{equation}  4\,a_{\vb{j}}\,a_{\vb{j}'}^{*}=\Lambda_{\vb{j},\vb{j}'}\qquad\text{for every }\{{\vb{j}},{\vb{j}}'\}\in E.
  \label{eqs_pol}
\end{equation}
It has a unique solution if and only if the graph $\mathcal{G}$ is connected. Without loss of generality, we restrict our attention to the minimal number of edges that guarantees a connected graph, i.e. $2^n-1$. For convenience, we relabel vertices such that vertex $v_{\vb{j}}$ is connected with vertex $v_{\vb{j}+1}$, for all $\vb{j}=0,\dots,2^n-1$.

To simplify notation, we redefine $\Lambda_{\vb{j}} := \Lambda_{\vb{j}-1,1}$. So, the solution to the system (\ref{eqs_pol}) is given by

\begin{equation}
  a_{\vb{j}}
  \;=\;
  a_0\,
  \frac{\Lambda_2^{*}\,\Lambda_4^{*}\,\cdots\,\Lambda_{\vb{j}}^{*}}
       {\Lambda_1\,\Lambda_3\,\cdots\,\Lambda_{\vb{j}-1}},
  \label{eq:closed-even}
\end{equation}
when $\vb{j}>0$ is even, and
\begin{equation}
  a_{\vb{j}}
  \;=\;
  \frac{1}{4a_0}
  \frac{\Lambda_1^{*}\,\Lambda_3^{*}\,\cdots\,\Lambda_{\vb{j}}^{*}}
       {\Lambda_2\,\Lambda_4\,\cdots\,\Lambda_{\vb{j}-1}},
  \label{eq:closed-odd}
\end{equation}
when $\vb{j}$ is odd. Here, we assumed that $a_0\geq0$ and $\Lambda_0=1$.

We emphasize here the crucial role that the polarization identity plays for reconstructing the quantum state, as it allows us to analytically solve the coupled system of equations (\ref{eqs_pol}). Otherwise, the method would require to solve a system of equations, \textit{procedure that has an exponential cost as a function of the number of qubits $n$}.

The analytical reconstruction described above is not unique, as multiple connected subgraphs of $\mathcal{G}$
can be used. Consequently, in the presence of statistical noise, different choices yield different reconstruction fidelities. To improve robustness, one might evaluate the fidelity obtained from each candidate subgraph and select the one that produces maximum fidelity. This selection step is a post-processing stage whose cost increases exponentially with the number of qubits.
$n$. In this work, we successfully applied post-processing up to 12 qubits on a standard personal computer, demonstrating feasibility at that scale. We emphasize that this post-processing stage is not required to implement our protocols, but it is recommended when the goal is to maximize the accuracy and overall quality of the reconstructed state.

As a final comment of this section, note from Eqs. (\ref{eqs_pol}), (\ref{eq:closed-even}) and (\ref{eq:closed-odd}) that, apparently, our tomographic methods work efficiently whenever we deal with a state $\ket{\psi}$ for which all entries in the computational basis are non-zero. However, the system of equations (\ref{eqs_pol}) can still be solved when $\ket{\psi}$ has some zero entries, as long as the associated graph is connected. Interestingly, from this fact we clearly see why when two non-consecutive coefficients $a_{\vb{j}}$ and $a_{\vb{j}+1}$ vanish, then the protocol from~\cite{goyeneche2015five} fails, as in this case the graph becomes disconnected, and consequently, the state cannot be univocally estimated.

For the $2n+1$ bases protocol, the reconstruction works up to $2^{n-1}-1$ zero entries in $\ket{\psi}$, whereas two zero entries in $\ket{\psi}$, associated with disconnected vertices, produce a disconnected graph for the 5 bases. To avoid this problem, one can reconstruct a state $\ket{\psi}'=U_1\otimes\dots\times U_n\ket{\psi}$, where $U_1$ to $U_n$ are suitable local unitary operations such that $\ket{\psi}'$ does not have zero entries. For example, a suitable choice for the GHZ state is a product of local Hadamard gates. Once $\ket{\psi}'$ is reconstructed, the target state is obtained as $\ket{\psi}=U^{\dag}_1\otimes\dots\times U^{\dag}_n\ket{\psi}'$, with a negligible additional propagation of errors introduced by the local gates $U_1,\dots,U_n$.

\section{Purity certification}\label{SM:purity}

    % \jcz{I think this part needs at least a paragraph to introduce why it is relevant. Something along the following lines} \bl{
    Our work introduces two tomographic protocols that provide a pure quantum state as an outcome. Despite the global state of an ideally isolated quantum system can be pure, systems are affected by several sources of errors and decoherence in practice, leading to mixed quantum states. An undesired situation for our tomographic protocols would be to prepare a quantum state that, in principle, should be nearly pure but in practice it turns out to be highly mixed, for an unknown reason. In order to prevent this kind of undesired situations,  in this Appendix we introduce a protocol that certifies that the collected statistical data is compatible with highly pure quantum states. This simple and efficient method was introduced in a previous work~\cite{goyeneche2015five} and can be generalized in the same way to the protocols introduced here. 
    
    The certification protocol is based on the following simple observation: entries of a quantum state $\rho\geq 0$ satisfy the following family of Cauchy-Schwarz inequalities:
    \begin{equation}\label{eq:CS_ineq_rho}
        \abs{\rho_{\vb{j}\vb{k}}}^2\leq \abs{\rho_{\vb{j}\vb{j}}}\abs{\rho_{\vb{k}\vb{k}}}.
    \end{equation}
Note that saturation of (\ref{eq:CS_ineq_rho}) for $e_{\vb{j},\vb{k}}$ defining edges of a connected graph $\mathcal{G}$ implies that $\rho$ is a rank one matrix, thus producing a pure quantum state. In other words, such a subset of saturating inequalities would imply that all columns of $\rho$ are linearly dependent. For both protocols introduced in this work, the graph is connected, so the certification protocol applies.

A key feature of our tomographic methods is that all entries $\rho_{\vb{j},\vb{k}}$ such that $e_{\vb{j},\vb{k}}\in\mathcal{G}$ are reconstructed, regardless of whether the prepared state $\rho$ is pure or not. In the case of a pure state, this information is enough to reconstruct the full density matrix $\rho$.

A quantity that reflects how close to pure is a given quantum state is provided by the following expression  
    \begin{equation}
\mathcal{P}=\sqrt{\sum_{e_{\vb{j},\vb{k}}\in\mathcal{G}}\qty(\abs{\rho_{jk}}^2-\rho_{jj}\rho_{kk})^2}.
    \end{equation}
Clearly, $\mathcal{P}=0$ if and only if $\rho$ is pure, as this implies the saturation of the Cauchy-Schwarz inequalities (\ref{eq:CS_ineq_rho}) for $e_{\vb{j},\vb{k}}\in\mathcal{G}$. However, we point out that $\mathcal{P}$ is not a monotonically increasing function of purity of the state. In fact, when $\rho$ is highly mixed, numerical simulations indicate that $\mathcal{P}$ can slightly increase when purity decreases. However, for highly pure states $\mathcal{P}$ is monotonically increasing with purity. 
% \rd{In Fig. \ref{fig:Purity_check} we shows the behavior of $\mathcal{P}$ for a density matrix $\rho=(1-\gamma)\op{\psi} + \gamma \mathbb{I}/d$ as a function of $\gamma$.} 
%     \begin{figure}[h!]
%         \centering
%         \includegraphics[width=0.8\linewidth]{Figures/purity_2n_2qb.png}
%         \caption{Function $\mathcal{P}$ as a function of the interpolating parameter $\gamma$, for a quantum state $\rho=(1-\gamma)\op{\psi} + \gamma \mathbb{I}/d$.\dg{Victor. Did you consider 5 or $2n+1$ bases here?  How many number of qubits? What is the state $|\psi\rangle$?}}
%         \label{fig:Purity_check}
%     \end{figure}

\section{Trapped ions experiment tomography results}\label{SM:IonTrap}
Fig. \ref{fig:density2q} and Fig. \ref{fig:density3q} show the full density matrices of the states studied in the trapped ion experiments, including the ideal states attained, the full tomographic reconstructions using Pauli Basis tomography, and the tomographic reconstruction using the 2$n$+1 basis method, which assumes a pure state. 
\begin{figure*}[t]
    \centering
    \includegraphics[width=\textwidth]{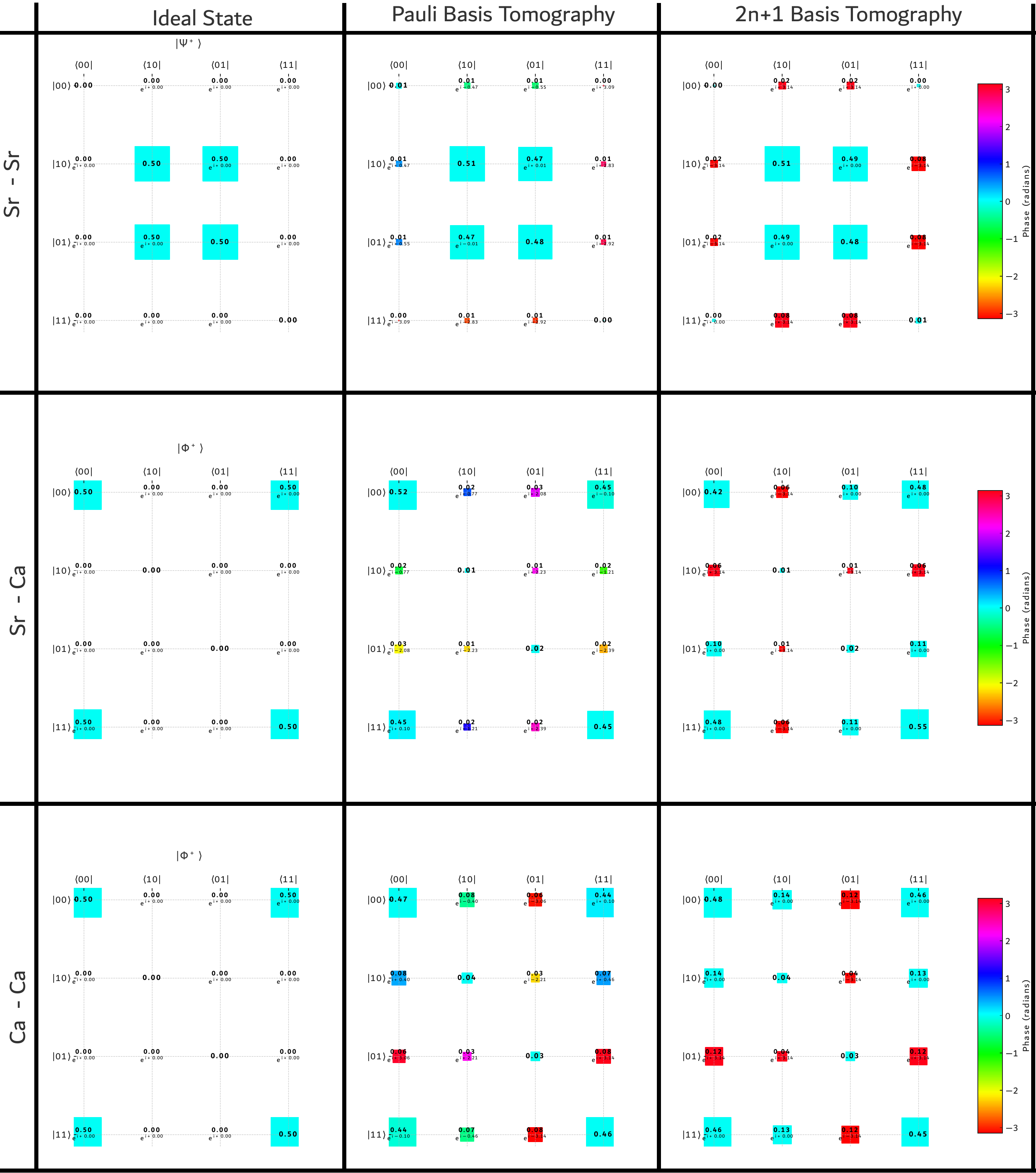}
    \caption{\justifying Density matrices of the different two-qubit states prepared in the trapped ion experiment. The magnitude of each matrix element is represented by the size of a square, while the phase of each element is plotted as a color.}
    \label{fig:density2q}  %%%%
\end{figure*}

\begin{sidewaysfigure*}[t]
    \centering
    \includegraphics[width=\textwidth]{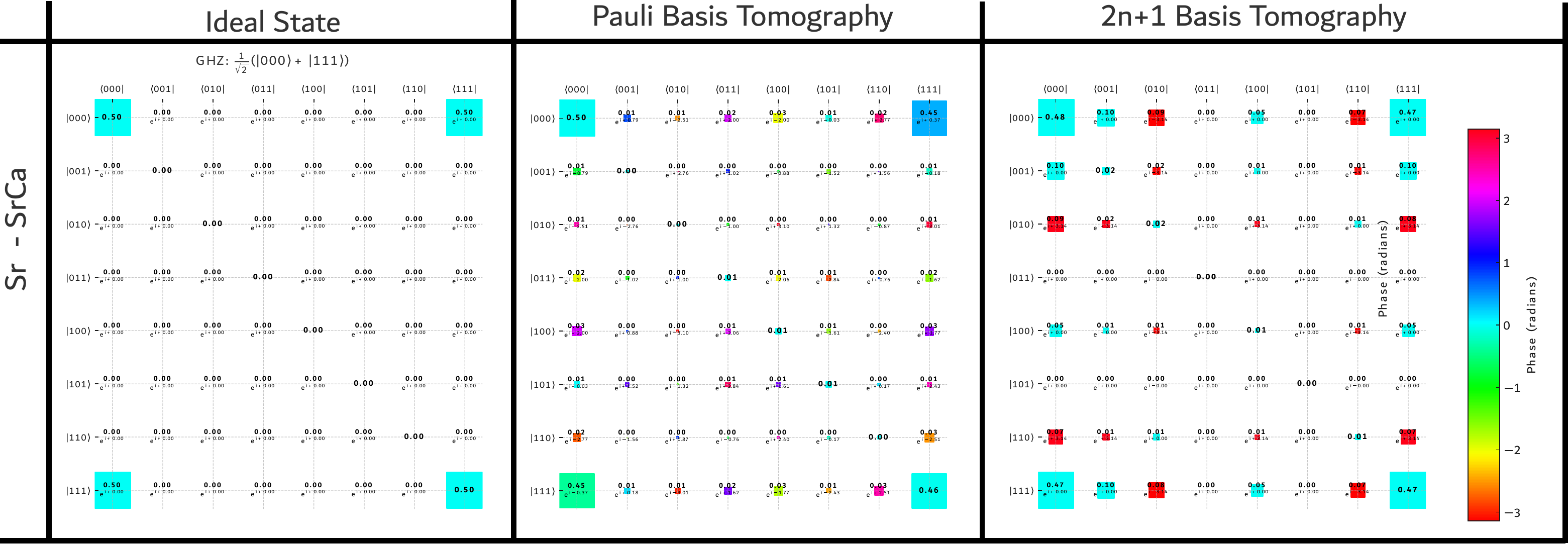}
    \caption{\justifying Density matrices of the three-qubit states prepared in the trapped ion experiment. The magnitude of each matrix element is represented by the size of a square, while the phase of each element is plotted as a color.}
    \label{fig:density3q}   %%%%
\end{sidewaysfigure*}

\section{Error propagation}
\label{apx:error_propagation}

    For a $n$-qubit quantum system with a quantum state given by
    \begin{equation}
        \ket{\psi}=\sum_{\vb{j}=0}^{2n - 1} a_{\vb{j}}\ket{\vb{j}}
    \end{equation}
    We will show how the error propagates when the estate estimation is performed. First, consider the polarization identity
    
        \begin{eqnarray}
            2a_{\tilde{\vb{k}}} a_{\vb{j}}^*  &=& \Lambda_{\tilde{\vb{k}}\vb{j}} \\  &=& \qty(p_{\tilde{\vb{k}}\vb{j}}^{+}-p_{\tilde{\vb{k}}\vb{j}}^{+})+i\qty(\tilde{p}_{\tilde{\vb{k}}\vb{j}}^{+}-\tilde{p}_{\tilde{\vb{k}}\vb{j}}^{+})\\ \label{eq:polar_id_re_im} &=&\Re{\Lambda_{\tilde{\vb{k}}\vb{j}}}+\Im{\Lambda_{\tilde{\vb{k}}\vb{j}}},
        \end{eqnarray}
        
        where ${p}_{\tilde{\vb{k}}\vb{j}}^{\pm}=\abs{a_{\tilde{\vb{k}}}\pm a_{\vb{j}}}$ and $\tilde{p}_{\tilde{\vb{k}}\vb{j}}^{\pm}=\abs{a_{\tilde{\vb{k}}}\pm ia_{\vb{j}}}$, which are measured directly by applying the projectors $\ket{\tilde{\vb{k}}}\pm\ket{\vb{j}}$ $\qty(\ket{\tilde{\vb{k}}}\pm i\ket{\vb{j}})$, to a quantum state. Then we calculate the $\Lambda_{\tilde{\vb{k}}\vb{j}}$ error propagation for the real and imaginary parts separately. For a function $f(\Vec{x})$ that depends on the variables $m$ $\{x_j\}_{j=1}^{m}$, to obtain an error propagation formula, we apply the following formula:
        
%        \begin{equation}\label{eq:error_prop}
%            (\Delta f)^2=\sum_{j=1}^{m}\qty(\pdv{f}{x_j})^2(\Delta x_j)^2 +\sum_{\substack{j,k=1\\j\neq k}}^{m}\qty(\pdv{f}{x_j})\qty(\pdv{f}{x_k})\text{Cov}(x_j,x_k)
%        \end{equation}

        \begin{align} 
        (\Delta f)^2 &= \sum_{j=1}^{m}\qty(\pdv{f}{x_j})^2 (\Delta x_j)^2 \notag \\
        &\quad + \sum_{\substack{j,k=1 \\ j \neq k}}^{m} \qty(\pdv{f}{x_j}) \qty(\pdv{f}{x_k}) \,\text{Cov}(x_j,x_k). \label{eq:error_prop}
        \end{align}

        In the case of Eq. \eqref{eq:polar_id_re_im}, let us start with the real part: 
         \begin{align}\label{eq:id_pol_errors}
             \qty(\Delta\Re{\Lambda_{\tilde{\vb{k}}\vb{j}}})^2&=\qty(\pdv{Re{\Lambda_{\tilde{\vb{k}}\vb{j}}}}{p_{\tilde{\vb{k}}\vb{j}}^{+}})^2\qty(\Delta p_{\tilde{\vb{k}}\vb{j}}^+)^2
             \notag \\
             &+\qty(\pdv{Re{\Lambda_{\tilde{\vb{k}}\vb{j}}}}{p_{\tilde{\vb{k}}\vb{j}}^{-}})^2\qty(\Delta 
             p_{\tilde{\vb{k}}\vb{j}}^-)^2 \nonumber
             \notag \\
             &+2 \qty(\pdv{Re{\Lambda_{\tilde{\vb{k}}\vb{j}}}}{p_{\tilde{\vb{k}}\vb{j}}^{+}})\qty(\pdv{Re{\Lambda_{\tilde{\vb{k}}\vb{j}}}}{p_{\tilde{\vb{k}}\vb{j}}^{-}}) \text{Cov}\qty(p_{\tilde{\vb{k}}\vb{j}}^+, p_{\tilde{\vb{k}}\vb{j}}^-) 
         \end{align}
         
         Note that 
         $$
         \qty(\pdv{Re{\Lambda_{\tilde{\vb{k}}\vb{j}}}}{p_{\tilde{\vb{k}}\vb{j}}^{+}})=1,\, \qty(\pdv{Re{\Lambda_{\tilde{\vb{k}}\vb{j}}}}{p_{\tilde{\vb{k}}\vb{j}}^{-}})=-1.
         $$
         Thus, we have
         %\dg{\textbf{Victor}. Please explain how to transform the above expression into this one.}
         
         \begin{align}             
         \qty(\Delta\Re{\Lambda_{\tilde{\vb{k}}\vb{j}}})^2=&\qty(\Delta p_{\tilde{\vb{k}}\vb{j}}^+)^2+\qty(\Delta p_{\tilde{\vb{k}}\vb{j}}^-)^2 
         \notag \\
         &-2\text{Cov}\qty(p_{\tilde{\vb{k}}\vb{j}}^+, p_{\tilde{\vb{k}}\vb{j}}^-)
         \end{align}
         
        Assuming a multinomial probability distribution we have        $$\qty(\Delta p_{\tilde{\vb{k}}\vb{j}}^\pm)^2=\frac{p_{\tilde{\vb{k}}\vb{j}}^\pm\qty(1-p_{\tilde{\vb{k}}\vb{j}}^\pm)}{{n}}$$
        where $n$ is the number of experiments measured on the projectors connecting $a_{\tilde{\vb{k}}}$ and $a_{\vb{j}}$. Also, we have $$\text{Cov}\qty(p_{\tilde{\vb{k}}\vb{j}}^+, p_{\tilde{\vb{k}}\vb{j}}^-)=-\frac{p_{\tilde{\vb{k}}\vb{j}}^+, p_{\tilde{\vb{k}}\vb{j}}^-}{{n}}.$$
        Therefore, the error for the real part of 
        \begin{equation}            \qty(\Delta\Re{\Lambda_{\tilde{\vb{k}}\vb{j}}})^2=\frac{\qty(p_{\tilde{\vb{k}}\vb{j}}^+ + p_{\tilde{\vb{k}}\vb{j}}^-)-\qty(p_{\tilde{\vb{k}}\vb{j}}^+-p_{\tilde{\vb{k}}\vb{j}}^-)^{2}}{n}.
        \end{equation}
        Errors estimation for the imaginary part $\Delta\Im{\Lambda_{\tilde{\vb{k}}\vb{j}}}$ is analogous to the real part but replacing $p_{\tilde{\vb{k}}\vb{j}}^\pm$ with  $\tilde{{p}}_{\tilde{\vb{k}}\vb{j}}^\pm$.\\ 

        On the other hand, given that $\Lambda_{\tilde{\vb{k}}\vb{j}}$ is complex
        \begin{equation}
            \Lambda_{\tilde{\vb{k}}\vb{j}}=\abs{\Lambda_{\tilde{\vb{k}}\vb{j}}}e^{i\beta_{\tilde{\vb{k}}\vb{j}}}
        \end{equation}
        It is also possible to estimate the error for the absolute value $\Lambda_{\tilde{\vb{k}}\vb{j}}$ and its complex phase. That is,

        %\begin{equation}            \abs{\Lambda_{\tilde{\vb{k}}\vb{j}}}=\sqrt{\Re{\Lambda_{\tilde{\vb{k}}\vb{j}}}^2+\Im{\Lambda_{\tilde{\vb{k}}\vb{j}}}^2} 
        %\end{equation}
        %Then, 
        %\begin{align*} \qty(\Delta\abs{\Lambda_{\tilde{\vb{k}}\vb{j}}})^2 &=            \qty(\pdv{\Lambda_{\tilde{\vb{k}}\vb{j}}}{\Re{\Lambda_{\tilde{\vb{k}}\vb{j}}}})^2\qty(\Delta\Re{\Lambda_{\tilde{\vb{k}}\vb{j}}})^2 \\&+ \qty(\pdv{\Lambda_{\tilde{\vb{k}}\vb{j}}}{\Im{\Lambda_{\tilde{\vb{k}}\vb{j}}}})^2\qty(\Delta\Im{\Lambda_{\tilde{\vb{k}}\vb{j}}})^2
        %\end{align*}

        %\begin{align*}            \qty(\Delta\abs{\Lambda_{\tilde{\vb{k}}\vb{j}}})^2&=\frac{1}{\abs{\Lambda_{\tilde{\vb{k}}\vb{j}}}^2}\left[\Re{\Lambda_{\tilde{\vb{k}}\vb{j}}}^2\qty(\Delta\Re{\Lambda_{\tilde{\vb{k}}\vb{j}}})^2+\right. \\ & \left.\Im{\Lambda_{\tilde{\vb{k}}\vb{j}}}^2\qty(\Delta\Im{\Lambda_{\tilde{\vb{k}}\vb{j}}})^2\right]
        %\end{align*}
        \begin{equation}
            \abs{\Lambda_{\tilde{\vb{k}}\vb{j}}}=\abs{c_{\tilde{\vb{k}}}}\abs{c_{\vb{j}}}=\sqrt{p_{\tilde{\vb{k}}}}\sqrt{p_{\vb{j}}}
        \end{equation}
         where  $p_{\vb{j}}:=\abs{c_{\vb{j}}}^2$. Then,
        \begin{align*}              \qty(\Delta\abs{\Lambda_{\tilde{\vb{k}}\vb{j}}})^2 &=\qty(\pdv{\abs{\Lambda_{\tilde{\vb{k}}\vb{j}}}}{p_{\tilde{\vb{k}}}})^2\qty(\Delta p_{\tilde{\vb{k}}})^2+\qty(\pdv{\abs{\Lambda_{\tilde{\vb{k}}\vb{j}}}}{p_{\vb{j}}})^2\qty(\Delta p_{\vb{j}})^2 \\ &+2\qty(\pdv{\abs{\Lambda_{\tilde{\vb{k}}\vb{j}}}}{p_{\tilde{\vb{k}}}})\qty(\pdv{\abs{\Lambda_{\tilde{\vb{k}}\vb{j}}}}{p_{\vb{j}}})\text{Cov}\qty(p_{\tilde{\vb{k}}}p_{\vb{j}}),
        \end{align*}
        where we get
        \begin{equation}            \qty(\Delta\abs{\Lambda_{\tilde{\vb{k}}\vb{j}}})^2 = \frac{p_{\tilde{\vb{k}}}+p_{\vb{j}}-4p_{\tilde{\vb{k}}}p_{\vb{j}}}{n}
        \end{equation}
        
        Now, for the phase
        \begin{equation}
            \beta_{\tilde{\vb{k}}\vb{j}}=\arctan{\qty(\frac{\Im{\Lambda_{\tilde{\vb{k}}\vb{j}}}}{\Re{\Lambda_{\tilde{\vb{k}}\vb{j}}}})}
        \end{equation}
        with a variation given by
        \begin{align}
            \qty(\Delta\beta_{\tilde{\vb{k}}\vb{j}})^2 &=
            \qty(\pdv{\beta_{\tilde{\vb{k}}\vb{j}}}{\Re{\Lambda_{\tilde{\vb{k}}\vb{j}}}})^2\qty(\Delta\Re{\Lambda_{\tilde{\vb{k}}\vb{j}}})^2+ \\ &
            \qty(\pdv{\beta_{\tilde{\vb{k}}\vb{j}}}{\Im{\Lambda_{\tilde{\vb{k}}\vb{j}}}})^2\qty(\Delta\Im{\Lambda_{\tilde{\vb{k}}\vb{j}}})^2
        \end{align}
        Resulting in
        \begin{align*}
            \qty(\Delta\beta_{\tilde{\vb{k}}\vb{j}})^2 &=\frac{1}{\abs{\Lambda_{\tilde{\vb{k}}\vb{j}}}^4}\left[\Im{\Lambda_{\tilde{\vb{k}}\vb{j}}}^2\qty(\Delta\Re{\Lambda_{\tilde{\vb{k}}\vb{j}}})^2\right. \\ & + \left. \Re{\Lambda_{\tilde{\vb{k}}\vb{j}}}^2\qty(\Delta\Im{\Lambda_{\tilde{\vb{k}}\vb{j}}})^2\right]
        \end{align*}

        Now, consider the particular case of a real quantum state $|\psi\rangle$, this is, $a_{\vb{j}}\in\mathbb{R}$ for every $\vb{j}\in \qty{\vb{0},\dots, 2^{n}-1}$. In this case, we have $\tilde{p}_{\tilde{\vb{k}}\vb{j}}^{+}=\tilde{p}_{\tilde{\vb{k}}\vb{j}}^{-}=a_{\tilde{\vb{k}}}^2+a_{\vb{j}}^2$ and
        \begin{equation}
            \qty(\Delta\beta_{\tilde{\vb{k}}\vb{j}})^2=\frac{2\tilde{p}_{\tilde{\vb{k}}\vb{j}}^{+}}{n\qty(p_{\tilde{\vb{k}}\vb{j}}^{+}-p_{\tilde{\vb{k}}\vb{j}}^{-})^2}.
        \end{equation}
        Moreover, if we consider that all the coefficients $a_{\vb{j}}=1/\sqrt{2^n} \forall\, \vb{j}\in \qty{\vb{0},\dots, 2^{n}-1}$. Then, we have $p_{\tilde{\vb{k}}\vb{j}}^{+}=4a_{\vb{j}}^{2}$ and $p_{\tilde{\vb{k}}\vb{j}}^{-}=0$. In this case, the variance in the phase is 
        \begin{equation}
            \qty(\Delta\beta_{\tilde{\vb{k}}\vb{j}})^2=\frac{1}{4n\abs{a_{\vb{j}}}^{2}}
        \end{equation}

        Now, consider
        \begin{align*}
            2a_{\tilde{\vb{k}}} a_{\vb{j}}^*=\Lambda_{\tilde{\vb{k}}\vb{j}}\\
            a_{\vb{j}}=\frac{\Lambda_{\tilde{\vb{k}}\vb{j}}^{*}}{2a_{\tilde{\vb{k}}}^{*}}
        \end{align*}    

        Let us show two cases. First if $a_{\tilde{\vb{k}}}\in\mathbb{R}$ Then
        \begin{align*}          
        \abs{a_{\vb{j}}}e^{i\alpha_{\vb{j}}}&=\frac{
        \abs{\Lambda_{\tilde{\vb{k}}\vb{j}}
        }e^{-i\beta_{\tilde{\vb{k}}\vb{j}}}}{2a_{\tilde{\vb{k}}}}.\\
        \end{align*}  
        Then, we have
        \begin{align*}
             \alpha_{\vb{j}}&=-\beta_{\tilde{\vb{k}}\vb{j}}\\
            \qty(\Delta\alpha_{\vb{j}})^2&=\qty(\Delta\beta_{\tilde{\vb{k}}\vb{j}})^2\\
        \end{align*}
        This is the case for all any coeficient directly connected to $a_{\tilde{\vb{k}}}$ through the projectors. 
        Moreover, for a coefficient $a_{\vb{l}}\in \mathbb{C}$ connected to $a_{\vb{j}}$ but not directly to $a_{\tilde{\vb{k}}}$ we have
        \begin{align*}
           2 a_{\vb{j}}a_{\vb{l}}^{*}&=\Lambda_{\vb{j}\vb{l}} \\
           \abs{a_{\vb{l}}}e^{i\alpha_{\vb{l}}}&=\frac{\abs{\Lambda_{\vb{j}\vb{l}}}}{\abs{a_{\vb{j}}}}e^{i(\alpha_{\vb{j}}-\beta_{\vb{j}\vb{l}})}
        \end{align*}
        And the error of the phase is given by:
        \begin{align*}           \qty(\Delta\alpha_{\vb{l}})^2=\qty(\Delta\alpha_{\vb{j}})^2+\qty(\Delta\beta_{\vb{j}\vb{l}})^2\\
        \qty(\Delta\alpha_{\vb{l}})^2=\qty(\Delta\beta_{\tilde{\vb{k}}\vb{j}})^2+\qty(\Delta\beta_{\vb{j}\vb{l}})^2\\
        \end{align*}  

        In this case, we can see that error propagation grows with the distance between the elements, in other words, the number of polarization identities we have to apply to find $\alpha_{\vb{l}}$. Also,  as $ \qty(\Delta\alpha_{\vb{j}})^2\propto \abs{\Lambda_{\tilde{\vb{k}}}\vb{j}}^{-4}$ the error grows when neither $p_{\tilde{\vb{k}}}$ or $p_{\vb{j}}$ are small. Therefore, to minimize error propagation it is recommended to find the shortest paths between element $a_{\vb{l}}$ and $a_{\tilde{\vb{k}}}$ that avoid elements with low probability.

\hfill

\bibliography{references}

\newpage

\end{document}